\documentclass[a4paper,UKenglish,cleveref, autoref, thm-restate]{lipics-v2021}

\hideLIPIcs  

\usepackage{booktabs}
\usepackage{amssymb}
\usepackage{amsmath}
\usepackage{amsthm}
\usepackage{mathtools}

\usepackage{cleveref}
\usepackage{subcaption}
\usepackage[textwidth=1.5cm]{todonotes}
\usepackage{tabularx}

{%
  \addtocounter{theorem}{-1}
  \endgroup
}%

{%
  \addtocounter{theorem}{-1}
  \endgroup
}%

\usepackage{tikz}
\usetikzlibrary{
    arrows,
    arrows.meta,
    decorations.pathmorphing,
    decorations.pathreplacing,
    calligraphy,
    positioning,
    matrix, 
    patterns,
    automata,
    calc,
    tikzmark,
    backgrounds,
    shapes
}

\tikzset{
    label distance=-1mm,
    world/.style = {
        rounded rectangle,
        draw,
        minimum width = 2em,
        minimum height = 2em,
        inner sep = 0.33em,
        },
    smallWorld/.style = {minimum height = 1.66em},
    submodel/.style = {line width=1.2pt}
}

\newcommand{\CTLsubmodel}[1]{\protect\ensuremath{\normalfont\text{E-}\Submodels(#1)}}
\newcommand{\CTLsubmodelPlain}{\protect\ensuremath{\normalfont\text{E-}\Submodels}}
\newcommand{\Del}[1]{\protect\ensuremath{\mathsf{Del}#1}}
\newcommand{\DelP}{\protect\ensuremath{\mathsf{DelP}}}
\newcommand{\DelNP}{\protect\ensuremath{\mathsf{DelNP}}}
\newcommand{\Ptime}{\protect\ensuremath{\mathsf{P}}}
\newcommand{\NPtime}{\protect\ensuremath{\mathsf{NP}}}

\newcommand{\Prop}{\protect\ensuremath{\mathrm{PROP}}}
\newcommand{\Sol}{\protect\ensuremath{\mathop\mathrm{Sol}}}
\newcommand{\E}{\protect\ensuremath{\mathcal{E}}}
\newcommand{\M}{\protect\ensuremath{\mathcal{M}}}
\newcommand{\ALL}{\protect\ensuremath{\mathcal{ALL}}}

\newcommand{\EX}{\protect\ensuremath{\mathop\mathsf{EX}}}
\newcommand{\EG}{\protect\ensuremath{\mathop\mathsf{EG}}}
\newcommand{\EF}{\protect\ensuremath{\mathop\mathsf{EF}}}
\newcommand{\EU}{\protect\ensuremath{\mathbin\mathsf{EU}}}
\newcommand{\ER}{\protect\ensuremath{\mathbin\mathsf{ER}}}
\newcommand{\AX}{\protect\ensuremath{\mathop\mathsf{AX}}}
\newcommand{\AG}{\protect\ensuremath{\mathop\mathsf{AG}}}
\newcommand{\AF}{\protect\ensuremath{\mathop\mathsf{AF}}}
\newcommand{\AU}{\protect\ensuremath{\mathbin\mathsf{AU}}}
\newcommand{\AR}{\protect\ensuremath{\mathbin\mathsf{AR}}}
\newcommand{\CTL}{\protect\ensuremath{\mathsf{CTL}}}

\newcommand{\HAMPATH}{{\normalfont\textsc{HAMPATH}}}
\newcommand{\SAT}{{\normalfont\textsc{SAT}}}
\newcommand{\pmRed}{\protect\ensuremath{\leq_m^{\Ptime}}}
\newcommand{\ExtendSolE}[1]{{\normalfont\textsc{ExtendSol}\_#1}}
\newcommand{\ExistE}[1]{{\normalfont\textsc{Exist}\_#1}}
\newcommand{\ExtendSol}{{\normalfont\textsc{ExtendSol}}}
\newcommand{\Exist}{\ensuremath{\exists}}
\newcommand{\Ext}{\normalfont\textsc{Ext}}
\newcommand{\Submodel}{\normalfont\textsc{Submodel}}
\newcommand{\Submodels}{\normalfont\textsc{Submodels}}
\newcommand{\ExtendSubmodel}[1]{{\Ext\normalfont\Submodel(#1)}}
\newcommand{\ExtendSubmodelplain}{{\Ext\normalfont\Submodel}}
\newcommand{\ExistSubmodel}[1]{{\Exist\normalfont\Submodel(#1)}}
\newcommand{\ExistSubmodelplain}{{\Exist\normalfont\Submodel}}

\usepackage{xifthen}

\newcommand{\envcite}[2][]{\ifthenelse{\isempty{#1}}{\cite{#2}}{\cite{#2}, #1}}

\title{Submodel Enumeration for CTL Is Hard} 

\author{Nicolas Fröhlich}{Leibniz Universität Hannover, Appelstrasse 9a, 30167 Hannover, Germany}{nicolas.froehlich@thi.uni-hannover.de}{https://orcid.org/0009-0003-5413-1823}{Funded by the German Research Foundation (DFG) under the project number ME4279/3-1.}

\author{Arne Meier}{Leibniz Universität Hannover, Appelstrasse 9a, 30167 Hannover, Germany}{meier@thi.uni-hannover.de}{https://orcid.org/0000-0002-8061-5376}{Partially funded by the German Research Foundation (DFG) under the project number ME4279/3-1.}

\authorrunning{N. Fröhlich and A. Meier} 

\Copyright{Nicolas Fröhlich and Arne Meier} 



\category{} 

\relatedversion{} 

\nolinenumbers 

\EventEditors{John Q. Open and Joan R. Access}
\EventNoEds{2}
\EventLongTitle{42nd Conference on Very Important Topics (CVIT 2016)}
\EventShortTitle{CVIT 2016}
\EventAcronym{CVIT}
\EventYear{2016}
\EventDate{December 24--27, 2016}
\EventLocation{Little Whinging, United Kingdom}
\EventLogo{}
\SeriesVolume{42}
\ArticleNo{23}

\begin{document}

\maketitle

\begin{abstract}
    Expressing system specifications using Computation Tree Logic (CTL) formulas, formalising programs using Kripke structures, and then model checking the system is an established workflow in program verification and has wide applications in AI.
    In this paper, we consider the task of model enumeration, which asks for a uniform stream of output systems that satisfy the given specification.
    We show that, given a CTL formula and a system (potentially falsified by the formula), enumerating satisfying submodels is always hard for CTL---regardless of which subset of CTL operators is considered.
    As a silver lining on the horizon, we present fragments via restrictions on the allowed Boolean functions that still allow for fast enumeration.
\end{abstract}

\section{Introduction}
In artificial intelligence, temporal logic is used as a formal language to describe and reason about the temporal behaviour of systems and processes~\cite{aitl00,DBLP:books/lib/Berard01}.
One of the key applications of temporal logic in artificial intelligence is the formal specification and verification of temporal properties of software systems, such as real-time systems~\cite{DBLP:journals/csur/BelliniMN00,DBLP:journals/fcsc/Konur13,BLOM199635}, reactive systems~\cite{finger1993metatem}, and hybrid systems~\cite{da2021symbolic}.
Temporal logic can be used to specify the desired behaviour of these systems and to check that systems of that kind satisfy the specified properties.
This task is known as \emph{model checking} (MC) and is one of the most important reasoning tasks~\cite{DBLP:conf/aiml/Schnoebelen02}. 
In this context, the search for satisfying submodels is a useful approach to debugging faulty systems.

One of the central temporal logics for which the model checking problem is efficiently solvable (more precisely, the problem is complete for polynomial time) is the Computation Tree Logic CTL~\cite{DBLP:books/daglib/0007403-2}.
The logic is often used in the context of program verification and, accordingly, is well suited to our study.
CTL formulas enrich classical propositional logic with a variety of modal operators (next, until, global, future, release) that combine with so-called path quantifiers (existential and universal) to form CTL operators.

Kripke structures, which model the software system of interest, are essentially labelled directed graphs that have a total transition relation~\cite{kripke63}.
A submodel of a Kripke structure is defined with respect to all possible subsets for the state set and transition relation.
Note that the subset for the transition relation must still be total, otherwise it is not a valid submodel.
More formally, the problem we are interested in is defined as follows.
For a Kripke structure $\M$ and a CTL formula $\varphi$, list of all submodels $\M'$ of $\M$ such that $\M'$ satisfies $\varphi$.
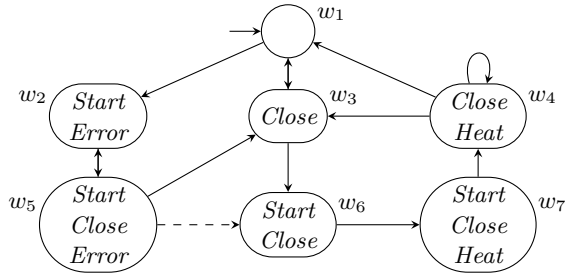
\begin{figure}
    \centering
    \begin{tikzpicture}[x=1.25cm, y=0.9cm]
        \node[world, label={ 5:\small $w_1$}, align=center, font=\small\itshape] (w1) at (0,1.25) {};
        \node[world, label={ 175:\small $w_2$}, align=center, font=\small\itshape] (w2) at (-2,0) {Start\\ Error};
        \node[world, label={ 5:\small $w_3$}, align=center, font=\small\itshape] (w3) at (0,0) {Close};
        \node[world, label={ 5:\small $w_4$}, align=center, font=\small\itshape] (w4) at (2,0) {Close\\ Heat};
        \node[world, label={ 175:\small $w_5$}, align=center, font=\small\itshape] (w5) at (-2,-1.6) {Start\\ Close\\ Error};
        \node[world, label={ 5:\small $w_6$}, align=center, font=\small\itshape] (w6) at (0,-1.6) {Start\\ Close};
        \node[world, label={ 5:\small $w_7$}, align=center, font=\small\itshape] (w7) at (2,-1.6) {Start\\ Close\\ Heat};
        
        \draw[stealth-] (w1.west) -- ++(-0.33,0);
        \draw [-stealth] (w1) -- (w2);
        \draw [-stealth] (w1) -- (w3);
        \draw [-stealth] (w2) -- (w5);
        \draw [-stealth] (w3) -- (w1);
        \draw [-stealth] (w3) -- (w6);
        \draw [-stealth] (w4) -- (w1);
        \draw [-stealth] (w4) -- (w3);
        \draw [-stealth] (w4) edge[loop above, >=Stealth, looseness = 6] (w4);
        \draw [-stealth] (w5) -- (w2);
        \draw [-stealth] (w5) -- (w3);
        \draw [-stealth] (w6) -- (w7);
        \draw [-stealth] (w7) -- (w4);

        \draw [-stealth, dashed] (w5) -- (w6);
    \end{tikzpicture}
    \caption{Kripke model of the microwave oven example. While the structure containing the dashed edge $(w_5, w_6)$ does not satisfy the constraint $\varphi = \AG(\textit{Error} \to \lnot \textit{Heat} \AU \lnot \textit{Start})$, the submodel without the edge does.}
    \label{fig:ex:microwave}
\end{figure}
Let us illustrate the idea with an example.
\begin{example}
    Consider the Kripke model $\M$ shown in \Cref{fig:ex:microwave}, which models the behavior of a microwave oven, a well-known example from \cite{DBLP:books/daglib/0007403-2}.
    Next, consider the constraint $\varphi = \AG(\textit{Error} \to \lnot \textit{Heat} \AU \lnot \textit{Start})$, which says that any path starting in a world labeled with the \textit{Error} proposition must first reach a world where $\lnot Start$ holds, before \textit{Heat} becomes true.
    With the dashed edge from $w_5$ to $w_6$ this constraint obviously does not hold in $\M$.
    In contrast, the submodel $\M'$ of $\M$ without the dashed edge satisfies $\varphi$.
    Thus, as an automated repair or a suggestion in debugging one might want to consider the submodels of $\M$.
\end{example}

Also other areas of research benefit from this kind of approach as follows.
For bounded model checking~\cite{DBLP:journals/ac/BiereCCSZ03}, the size of the state space can easily become very large for complex systems. For such systems, Biere et al.\ suggest combining model checking with SAT solvers, which allow faster exploration of the state space. 
Similarly, Gupta~et~al.~\cite{DBLP:conf/fmcad/GuptaYAG00} showed that similar things have been observed in the context of BDD-based symbolic algorithms for image processing.
While one might think that the work of Sullivan et al.~\cite{DBLP:conf/icfem/SullivanMK19} is closely related to our setting, the authors work with propositional logic in the setting of the specification language Alloy, which is based on first-order logic. 
This is somewhat different from us, as we work with CTL. However, there is work on CTL-live model checking for first-order logic validity checking~\cite{DBLP:conf/fmcad/VakiliD14}, but this would be a different direction to our approach.
Lauri and Dutta~\cite{DBLP:conf/aaai/LauriD19}, following an ML perspective, devise an ML framework that attempts to shrink the search space and augment the solver with some help.
Recently, this topic has been investigated in the context of plain modal logic~\cite{DBLP:conf/aiml/FrohlichM22}.

Classically, the task of enumerating models is very different from counting the number of existing models or deciding on the existence of such models.
Although enumeration algorithms are usually exponential time algorithms, this does not preclude practical applications.
In addition to theoretical studies, there are a number of application scenarios~\cite{exexalg10}, e.g., recommender systems~\cite{DBLP:journals/aim/FriedrichZ11}, ASP~\cite{DBLP:conf/aiia/AlvianoD16}, or ML~\cite{DBLP:conf/aaai/LauriD19}. 
The formal foundations were originally laid by Johnson~et~al.~\cite{JohnsonYannakakisPapadimitriou1988}.
Intuitively, an enumeration algorithm is deterministic and produces a uniform stream of output solutions avoiding duplicates.
This solution flow is mathematically modelled by the notion of the \emph{delay}, i.e., an upper bound on the elapsed time between printing two successive solutions (or the time before the first, resp., after the last, solution is returned).
In~\cite{DBLP:journals/dam/CreignouKPSV19}, Creignou~et~al. introduced a framework for intrinsically hard enumeration problems.
Here, the polynomial hierarchy and the concept of oracle machines have been utilised to present notions that allow for proving intractability bounds for enumeration problems.
The complexity class $\DelP$ describes ``efficient'' enumeration, that is, a delay that is polynomially bounded in the input length, while the complexity class $\DelNP$ contains intractable and, accordingly, difficult enumeration problems.
Solutions to instances of problems in $\DelNP$ cannot efficiently be produced unless the (classical) complexity classes $\Ptime$ and $\NPtime$ coincide.

While the tractability of MC for CTL formulas is known to be $\Ptime$-complete, the complexity of enumerating satisfying submodels is still open.

\paragraph*{Contributions.} In this paper, we fill this gap and present a thorough study of the complexity of the submodel enumeration problem in the context of CTL.
We will see that in general the problem is complete for the class $\DelNP$; so it is reasonable to consider restrictions that aim for tractable enumeration cases.
However, our answer in this direction is on the negative side, showing that any restriction on the CTL operator side does not allow faster enumeration algorithms (assuming $\Ptime\neq\NPtime$).
Finally, we identify some further Boolean restrictions that still allow for $\DelP$ algorithms.

\paragraph*{Related works.}
The work of Schnoebelen~(\cite{DBLP:conf/aiml/Schnoebelen02}) considers the classical model checking question for temporal logics.
There is a study on the complexity of fragments of the model checking problem~\cite{DBLP:journals/acta/KrebsMM19} for CTL, but it has no direct impact on our results as the problems are situated in $\Ptime$ (or classes within).

\paragraph*{Organisation.}
At first, we introduce the necessary preliminaries on temporal logic and enumeration complexity.
Then we prove our dichotomy theorem.
Finally, we conclude.

\section{Preliminaries}
In this section, we assume basic familiarity with computational complexity~\cite{DBLP:books/daglib/0018514} and will make use of the framework for hard enumeration problems~\cite{DBLP:journals/dam/CreignouKPSV19}.
Furthermore, we will define the temporal logic $\CTL$, introduce submodels, and enumeration complexity.

\paragraph*{Computational Tree Logic.} 
We follow standard notation of model checking~\cite{DBLP:books/daglib/0007403-2}.
Let $\Prop$ be an infinite, countable set of propositions.
The set of well-formed $\CTL$ formulas is defined with the following BNF
\[
    \varphi \Coloneqq
    \top \mid
    p \mid
    \lnot \varphi \mid
    \varphi \lor \varphi \mid
    \varphi \land \varphi \mid
    \mathop{\mathcal P\mathcal T} \varphi \mid
    \varphi \mathbin{\mathcal P\mathcal T'} \varphi, 
\]
where $p \in \Prop, \mathcal P \in \{\mathsf{E}, \mathsf{A}\}, \mathcal T \in \{\mathsf{X}, \mathsf{F}, \mathsf{G}\}, \mathcal T' \in \{\mathsf{U}, \mathsf{R}\}$.
This results in ten $\CTL$ operators, consisting of six unary operators $\EX$, $\EX$, $\EF$, $\AF$, $\EG$, $\AG$ and four binary operators $\EU$, $\AU$, $\ER$, $\AR$.
The set $\ALL$ contains all ten $\CTL$ operators.
We will call $\land,\lor,\lnot$ Boolean connectors.

In the following we defined a special version of Kripke models.
\begin{definition}
    A \emph{rooted Kripke model} is a tuple $\M = (W, R, \eta, r)$ where
    \begin{itemize}
        \item $W$ is a non-empty set of \emph{worlds} (or \emph{states}),
        \item $R \subseteq W \times W$ is a total, binary transition relation on $W$
        \item $\eta \colon W \to 2^\Prop$ is an assignment function, that maps each world $w$ to a set $\eta(w)$ of propositions, and
        \item $r \in W$ is the \emph{root}.
    \end{itemize}
\end{definition}

\begin{definition}
    Let $\M = (W, R, \eta, r)$ be a rooted Kripke model.
    A \emph{path} $\pi$ in $\M$ is an infinite sequence of worlds $w_1,w_2,\dots $ such that $(w_i,w_{i+1}) \in R$ for all $i \geq 1$.
    We write $\pi[i]$ to denote the $i$th world on the path $\pi$.
    For a world $w \in W$ we define $\Pi(w) \coloneqq \{\,\pi \mid \pi[1] = w\,\}$ as the  (possibly infinite) set of all infinite paths of $\M$ starting with $w$.
\end{definition}
\begin{definition}
    Let $\M$ be a rooted Kripke model and $\varphi, \psi$ be CTL formulas.

    \begin{tabbing}
        \indent\=$\M, w \models \varphi \ER \psi$\hspace{0.33em} \= iff \= Rechts \kill
        \>$\M, w \models \top$\>\> always,\\
        \>$\M, w \models p$\> iff\> $p \in \eta(w)$ with $p \in \Prop$,\\
        \>$\M, w \models \lnot \varphi$\> iff\> $\M, w \not\models \varphi$,\\
        \>$\M, w \models \varphi \land \psi$\> iff \> $\M, w \models \varphi$ and $\M, w \models \psi$,\\
        \>$\M, w \models \varphi \lor \psi$\> iff\> $\M, w \models \varphi$ or $\M, w \models \psi$,\\
        \>$\M, w \models \EX \varphi$\> iff\> $\exists \pi \in \Pi(w): \M, \pi[2] \models \varphi$,\\
        \>$\M, w \models \AX \varphi$\> iff\> $\forall \pi \in \Pi(w): \M, \pi[2] \models \varphi$,\\
        \>$\M, w \models \EF \varphi$\> iff\> $\exists \pi \in \Pi(w)~\exists k \geq 1: \M, \pi[k] \models \varphi$,\\
        \>$\M, w \models \AF \varphi$\> iff\> $\forall \pi \in \Pi(w)~\exists k \geq 1: \M, \pi[k] \models \varphi$,\\
        \>$\M, w \models \EG \varphi$\> iff\> $\exists \pi \in \Pi(w)~\forall k \geq 1: \M, \pi[k] \models \varphi$,\\
        \>$\M, w \models \AG \varphi$\> iff\> $\forall \pi \in \Pi(w)~\forall k \geq 1: \M, \pi[k] \models \varphi$,\\
        \>$\M, w \models \varphi \EU \psi$\> iff\> $\exists \pi \in \Pi(w)~\exists k \geq 1: \M, \pi[k] \models \psi$\\
        \>\>\> and $\forall i < k: \M, \pi[i] \models \varphi$,\\
        \>$\M, w \models \varphi \AU \psi$\> iff\> $\forall \pi \in \Pi(w)~\exists k \geq 1: \M, \pi[k] \models \psi$\\
        \>\>\> and $\forall i < k: \M, \pi[i] \models \varphi$,\\
        \>$\M, w \models \varphi \ER \psi$\> iff\> $\exists \pi \in \Pi(w)~\forall k \geq 1: \M, \pi[k] \models \psi$\\
        \>\>\> or $\exists i < k: \M, \pi[i] \models \varphi$,\\
        \>$\M, w \models \varphi \AR \psi$\> iff\> $\forall \pi \in \Pi(w)~\forall k \geq 1: \M, \pi[k] \models \psi$\\
        \>\>\> or $\exists i < k: \M, \pi[i] \models \varphi$.
    \end{tabbing}
    Furthermore, $\bot \coloneqq \lnot \top$ is constant false.
    Also omit the root in $\M, r \models \varphi$ and just write $\M\models \varphi$ instead.
    A formula $\varphi$ is then said to be \emph{satisfied by model $\M$}, if $\M\models \varphi$ is true.
\end{definition}
\noindent Notice an observation regarding the semantics of $\CTL$.

\begin{observation}\label{ob:CTL_equis}
    The following equivalences hold:
    \begin{align*}
         & \EX \varphi \equiv \lnot \AX (\lnot \varphi), \AG \varphi \equiv \lnot \EF (\lnot \varphi), \EG \varphi \equiv \lnot \AF (\lnot \varphi) \\
         & \EG \varphi \equiv \bot \ER \varphi, \AG \varphi \equiv \bot \AR \varphi                                                                 \\
         & \EF \varphi \equiv \top \EU \varphi, \AF \varphi \equiv \top \AU \varphi                                                                 \\
         & \varphi \ER \psi \equiv \lnot (\lnot \varphi \AU \lnot \psi), \varphi \AR \psi \equiv \lnot (\lnot \varphi \EU \lnot \psi)
    \end{align*}
\end{observation}

Now, we formally introduce submodels of Kripke models.
Given two Kripke models  $\M = (W, R, \eta, r)$ and $\M' = (W', R', \eta, r)$.
We call $\M'$ a \emph{submodel (of $\M$)}, if $W' \subseteq W$, $R' \subseteq R$, and $R'$ is total.
For a function $f\colon A\to B$ we write $f\!\upharpoonright_{C}$, given $C\subseteq A$, for the restriction of $f$ to domain $C$.

\begin{definition}
    Let $\M = (W, R, \eta, r)$ be a Kripke model.
    $\M' = (W', R', \eta\!\upharpoonright_{W'}, r)$ is a \emph{connected submodel} of $\M$, denoted by $\M' \subseteq \M$, if
    (1.) $W' \neq \emptyset$,
    (2.) $\M'$ is a submodel of $\M$, and
    (3.) for all $w \in W'$ there exists a path $\pi \in \Pi(r)$ and $i \geq 1$ with $\pi[i] = w$.
\end{definition}
Clearly, worlds that violate (3.) cannot have influence on the satisfiability of CTL formulas.
Yet, an enumeration algorithm printing connected submodels could trivially be extended to include non-connected submodels.

Additionally we want to introduce an alternative notation for submodels, $\M' = \M - D$, with $D = (D_W, D_R)$ for $r\notin D_W$ a tuple consisting of a set of worlds and a set of tuples, and $W' = W \setminus D_W$ and  $R' = R \setminus D_R$, for $\M = (W, R, \eta, r)$ and $\M' = (W', R', \eta\!\upharpoonright_{W'}, r)$. 
Here, $D$ is called the set of \emph{deletions}.

A submodel $\M'$ is \emph{satisfying $\varphi$} if $\M' \models \varphi$.
The formula $\varphi$ is often omitted, if it can be deduced from the context.

\paragraph*{Enumeration Complexity.}
The Turing machine, as one of the standard machine models used in complexity theory, proves to be problematic for the setting of enumeration algorithms.
Its linear nature in accessing data prevents a polynomial delay when traversing exponentially large data sets, even if the actual data read is small.
As a result, random access machines (RAMs) are the common machine model of choice~\cite{DBLP:journals/eatcs/Strozecki19}.

\begin{definition}
    Let $\Sigma$ be a finite alphabet. An \emph{enumeration problem} is a tuple $\E = (I,\Sol)$, where
    \begin{itemize}
        \item $I \subseteq \Sigma^*$ is the set of \emph{instances},
        \item $\Sol \colon I \to \mathcal{P}(\Sigma^*)$ is a function that maps each instance $x \in I$ to a set of \emph{solutions (of $x$)}, and
        \item there exists a polynomial $p$ such that $\forall x\in I\;\forall y\in \Sol(x)$ we have that $|y|\leq p(|x|)$.
    \end{itemize}\label{def:enumeration problem}
\end{definition}
Note that sometimes one is interested in dropping the last requirement of the previous definition~\cite{DBLP:journals/eatcs/Strozecki19}.

\begin{definition}
    Let $\E = (I,\Sol)$ be an enumeration problem.
    An algorithm $\mathcal{A}$ is called an \emph{enumeration algorithm} for $\E$, if for every instance $x \in I$: $\mathcal{A}(x)$ terminates after a finite sequence of steps and $\mathcal{A}(x)$ prints exactly $\Sol(x)$ without duplicates, where $\mathcal{A}(x)$ denotes the \emph{computation of $\mathcal A$ on input~$x$}.
\end{definition}

We now define the mentioned delay of an enumeration algorithm.

\begin{definition}
    \label{def:delay}
    Let $\E = (I,\Sol)$ be an enumeration problem, $\mathcal{A}$ be an enumeration algorithm for $\E$, $x \in I$ be an instance and $n = |\Sol(x)|$ the number of solutions of $x$.
    We define the
    \begin{itemize}
        \item \emph{$i$th delay} of $\mathcal{A}(x)$ as the elapsed time between the output of the $i$th and $(i+1)$st solution of $\Sol(x)$,
        \item \emph{$0$th delay} as the \emph{precomputation time}, i.e, the elapsed time before the first output of $\mathcal A(x)$, and
        \item \emph{$n$th delay} as the \emph{postcomputation time}, i.e., the elapsed time after the last output of $\mathcal A(x)$ until it terminates.
    \end{itemize}
    We say that $\mathcal A$ has \emph{delay $f$}, for $f \colon \mathbb{N} \to \mathbb{N}$, if for all $x\in I$ and all $0\leq i \leq n$ the $i$th delay of $\mathcal A(x)$ is in $O(f(|x|))$.
\end{definition}

\paragraph*{Hard enumeration.}
We will shortly introduce the framework of hard enumeration by Creignou~et~al. \cite{DBLP:journals/dam/CreignouKPSV19}.
The idea is to analyse enumeration problems beyond polynomial delay by introducing a hierarchy of complexity classes similar to the polynomial-time hierarchy and reduction notions for enumeration problems.

We begin by defining two decision problems that naturally arise in the context of enumeration.
Let $\E = (I, \Sol)$ be an enumeration problem over the alphabet $\Sigma$.
The first decision problem $\ExistE{\E}$ asks, given an instance $x$, for the existence of any solutions, that is, $\Sol(x)$ is nonempty.

The second decision problem is concerned with obtaining new solutions.
This is the question whether, given an instance $x$ and a partial solution $y$, can we extend the partial solution by a word $y' \subseteq \Sigma^*$ such that $yy'$ is a solution of $\E$, where $yy'$ denotes the concatenation of $y$ and $y'$.
\begin{center}
    \begin{tabularx}{\columnwidth}{rX}
        \toprule
        \textbf{Problem:}  & $\ExtendSolE{\E}$                               \\
        \midrule
        \textbf{Input:}    & Instance $x$, partial solution $y$              \\
        \textbf{Question:} & Is there some $y'$ such that $yy' \in \Sol(x)$? \\
        \bottomrule
    \end{tabularx}
\end{center}

As mentioned before, we use RAMs instead of Turing machines in the context of enumeration complexity.
We now want to further extend the underlying machine model, by introducing decision oracles.
Classically, when analysing runtime, or in this case delay, algorithm calls to its oracle are always charged as a single step, regardless of the time the oracle takes.
Our machines can write into special registers and the oracle will consider these as well as all consecutive non-empty registers as its input.
A query to the oracle then occurs when the machine enters a special question state and will transition into a positive/negative state if the oracle answers ``yes''/``no''.
Now, start with enumeration complexity classes with oracles.
\begin{definition}[{\envcite[Def.~2]{DBLP:journals/dam/CreignouKPSV19}}]
    Let $\E$ be an enumeration problem, and $\mathcal C$ a decision complexity class. Then we say that $\E \in \Del{\mathcal{C}}$ if there is a RAM $M$ with oracle $L$ in $\mathcal C$ and a polynomial $p$, such that for any instance $x$, the RAM $M$ enumerates $\Sol(x)$ with delay $p(|x|)$.
    Moreover, the size of every oracle call is bound by $p(|x|)$.
\end{definition}
In this paper, the enumeration classes of interest are when $\mathcal C$ is either $\Ptime$, or $\NPtime$; so $\DelP$ and $\DelNP$.

The following \Cref{prop:ExtendSolInC-EInDelC} as well as \Cref{th:exist_NP-hard_E_DelNP-hard} are both simplified versions of results presented in~\cite{DBLP:journals/dam/CreignouKPSV19}.
While Creignou~et~al. considered the full polynomial hierarchies in their proofs, here we are only concerned with the $\Ptime$ and $\NPtime$ cases.

\begin{proposition}[{\envcite[Prop.~6]{DBLP:journals/dam/CreignouKPSV19}}]
    \label{prop:ExtendSolInC-EInDelC}
    Let $\E = (I, \Sol)$ be an enumeration problem and $\mathcal C \in \{\Ptime, \NPtime\}$.
    If $\ExtendSolE{\E} \in \mathcal C$ then $\mathcal E \in \Del{\mathcal C}$.
\end{proposition}
\Cref{prop:ExtendSolInC-EInDelC} allows for membership results for enumeration problems using the corresponding decision problem $\ExtendSol$.
This technique will prove particularly useful when showing membership in $\DelNP$, as constructing enumeration algorithms with oracles can be quite difficult.

We now give the necessary definitions to show hardness results for enumeration problems.
The first definition introduces yet another machine model, which can then be used to define a reduction from one enumeration problem to another.
\begin{definition}[{\envcite[Def.~6]{DBLP:journals/dam/CreignouKPSV19}}]
    Let $\E$ be an enumeration problem.
    An \emph{Enumeration Oracle Machine with an enumeration oracle $\E$}, abbreviated as EOM\_$\E$, is a RAM that has a sequence of new registers $A_e, O^e(0), O^e(1), \dots$ and a new instruction NOO (next oracle output).
    An EOM\_$\E$ is \emph{oracle-bounded}, if the size of all inputs to the oracle is at most polynomial in the size of the input to the EOM\_\E.
\end{definition}
Note that the sequence of registers as input is only necessary for EOM\_$\E$ that are not oracle-bounded, to allow input sizes larger than polynomial.
\begin{definition}[{\envcite[Def.~7]{DBLP:journals/dam/CreignouKPSV19}}]
    Let $\E = (I , \Sol)$ be an enumeration problem and $\rho_1, \rho_2, \dots$ be the run of an EOM\_$\E$ and assume that the $k$th instruction is NOO, that is, $\rho_k =$ NOO.
    Denote with $x_i$ the word stored in $O^e(0), O^e(1), \dots$ at step $i$.
    Let $K = \{\rho_j \in \{\rho_1, \dots, \rho_{k-1}\} \mid \rho_j = \text{NOO and } x_j = x_k\}$.
    Then the \emph{oracle output $y_k$ in $\rho_k$} is defined as an arbitrary $y_k \in \Sol(x_k)$ such that $y_k$ has not been the oracle output in any $\rho_j \in K$.
    If no such $y_k$ exists, then the oracle output in $\rho_k$ is undefined.

    On executing NOO in step $\rho_k$, if the oracle output $y_k$ is undefined, then register $A_e$ contains some special symbol in step $\rho_{k+1}$; otherwise it contains $y_k$.
\end{definition}
\begin{definition}[$D$-reductions]
    Let $\E$ and $\E'$ be enumeration problems.
    We say that $\E$ reduces to $\E'$ via \emph{$D$-reduction}, $\E \leq_D {\E'}$, if there is an oracle-bounded EOM\_${\E'}$ that enumerates $\E$ in $\DelP$ and is independent of the order in which the $\E'$-oracle enumerates its answers.
\end{definition}

The next result shows that one can use the decision problem $\ExistE{\E}$ to show hardness of the corresponding enumeration problem $\E$.

\begin{proposition}[{\envcite[Theorem 13]{DBLP:journals/dam/CreignouKPSV19}}]
    \label{th:exist_NP-hard_E_DelNP-hard}
    Let $\E = (I, \Sol)$ be an enumeration problem.
    If $\ExistE{\E}$ is $\NPtime$-hard, then $\E$ is $\DelNP$-hard via D-reductions.
\end{proposition}

Any following result that states $\DelNP$-hardness for an enumeration problem will be with respect to $D$-reductions.

\section{Complexity of Submodel Enumeration}
In this section, we will present our results regarding the submodel enumeration problem with respect to CTL formulas.
\begin{center}
    \begin{tabularx}{\columnwidth}{rX}
        \toprule
        \textbf{Problem:} & $\CTLsubmodelPlain$                                        \\
        \midrule
        \textbf{Input:}   & Kripke model $\M$, $\CTL$ formula $\varphi$                \\
        \textbf{Task:}    & Output all $\M'\subseteq\M$ such that $\M'\models\varphi$? \\
        \bottomrule
    \end{tabularx}
\end{center}
Let $\mathcal O$ be a set of $\CTL$ operators.
Then $\CTLsubmodel{\mathcal O}$ is $\CTLsubmodelPlain$ but only for $\CTL$ formulas using operators from $\mathcal O$ (besides any of the Boolean connectors).
The same applies to the next two auxiliary decision problems.
\begin{center}
    \begin{tabularx}{\columnwidth}{rX}
        \toprule
        \textbf{Problem:}  & \ExistSubmodelplain                                            \\
        \midrule
        \textbf{Input:}    & Kripke model $\M$, $\CTL$ formula $\varphi$                    \\
        \textbf{Question:} & Does $\M' \subseteq \M$ exist such that $\M' \models \varphi$? \\
        \bottomrule
    \end{tabularx}
\end{center}

\begin{center}
    \begin{tabularx}{\columnwidth}{rX}
        \toprule
        \textbf{Problem:}  & \ExtendSubmodelplain                               \\
        \midrule
        \textbf{Input:}    & Kripke model $\M$, $\CTL$ formula $\varphi$,       \\
                           & set of deletions $D$                               \\
        \textbf{Question:} & Does an extension $D' \supseteq D$ exist such that \\
                           & $\M - D' \models \varphi$?                         \\
        \bottomrule
    \end{tabularx}
\end{center}

The first result will show membership in the class $\DelNP$ for the unrestricted version and will make use of the auxiliary problem $\ExtendSubmodelplain$.

\begin{theorem}\label{th:CTL_InDelNP}
    $\CTLsubmodelPlain \in \DelNP$.
\end{theorem}
\begin{proof}
    The algorithm deciding $\ExtendSubmodelplain$ works as follows.
    For given Kripke model $\M = (W,R,\eta, r)$, $\CTL$ formula $\varphi$, and set of deletions $D = (D_W, D_R)$,  guess $W' \subseteq W$ and $R' \subseteq R$.
    Afterwards compute $D' \coloneqq (W' \cup D_W, R'\cup D_R)$ and accept if and only if $\M - D' \models \varphi$.

    For correctness, consider that if an extension $D'$ exists such that $\M - D' \models \varphi$, it can be computed by nondeterministically guessing the worlds and relations of that extension. 
    Guessing $W'$ and $R'$, computing $D'$ and checking if $\M - D' \models \varphi$ can all clearly be done in polynomial time (MC is in $\Ptime$ for $\CTL$).
    By \Cref{prop:ExtendSolInC-EInDelC} this is sufficient to prove that $\CTLsubmodelPlain \in \DelNP$.
\end{proof}

\paragraph*{Fragment AG.}
We will show hardness by relating submodels to assignments of propositional formulas, such that a submodel is satisfying, if and only if the corresponding assignment satisfies the given propositional formulas.
Formally this is a reduction from the well-known $\NPtime$-complete problem $\SAT$~\cite{DBLP:conf/stoc/Cook71,lev73}.
\begin{definition}\label{def:M(varphi)}
    Let $\varphi$ be a pro\-po\-si\-tio\-nal formula with $\Prop(\varphi) = \{x_1, x_2, \dots, x_n\}$.
    We define the Kripke model $\M(\varphi) \coloneqq (W, R, \eta, w_0)$ as follows:
    \begin{align*}
        W \coloneqq{}           & \{w_0\} \cup \{w_i^0, w_i^1 \mid 1 \leq i \leq n\}                 \\
        R \coloneqq{}           & \{(w_0, w_1^k) \mid k \in \{0, 1\}\}                               \\
                                & \cup \{(w_{i}^k, w_{i+1}^l) \mid k, l \in \{0, 1\}, 1 \leq i < n\} \\
                                & \cup \{(w_n^k, w_n^k) \mid k \in \{0, 1\}\}                        \\
        \eta(w_i^k) \coloneqq{} & \{x_i, x_i^k\} \text{ for } 1 \leq i \leq n, k \in \{0,1\}
    \end{align*}
\end{definition}

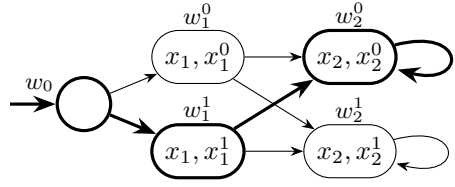
\begin{figure}
    \centering
    \begin{tikzpicture}[fix/.style={},y=1.25cm]
        \node[world,submodel,label={170:\small$w_0$}] at (-.5,0) (s) {};
        \node[world,fix,label={90:\small$w_1^0$}] at (1,.5) (x_1^0) {$x_1,x_1^0$};
        \node[world,fix,submodel,label={90:\small$w_1^1$}] at (1,-.5) (x_1^1) {$x_1,x_1^1$};
        \node[world,fix,submodel,label={90:\small$w_2^0$}] at (3,.5) (x_2^0) {$x_2,x_2^0$};
        \node[world,fix,label={90:\small$w_2^1$}] at (3,-.5) (x_2^1) {$x_2,x_2^1$};

        \draw[Stealth-,submodel] (s.west) -- ++(-.6,0);
        \draw[-Stealth] (s) -- (x_1^0);
        \draw[-Stealth,submodel] (s) -- (x_1^1);
        \draw[-Stealth] (x_1^0) -- (x_2^0);
        \draw[-Stealth] (x_1^0) -- (x_2^1);
        \draw[-Stealth,submodel] (x_1^1) -- (x_2^0);
        \draw[-Stealth] (x_1^1) -- (x_2^1);
        \draw[-Stealth,submodel] (x_2^0) edge[loop right, >=Stealth] (x_2^0);
        \draw[-Stealth] (x_2^1) edge[loop right, >=Stealth] (x_2^1);
    \end{tikzpicture}
    \caption{Kripke model $\M((x_1 \land \lnot x_2) \lor (\lnot x_1 \land x_2))$ and a submodel $\M'$ of $\M$ in bold.}
    \label{fig:ex:M(varphi)}
\end{figure}

\Cref{fig:ex:M(varphi)} depicts such a Kripke model together with one of its submodels.
Notice that the formula $\varphi = (x_1 \land \lnot x_2) \lor (\lnot x_1 \land x_2)$ of $\M$ is satisfied given the assignment $\mathfrak I(x_1) = 1, \mathfrak I(x_2) = 0$, which the submodel $\M'$ ``encodes'' by containing the worlds $w_1^1, w_2^0$ and not $w_1^0, w_2^1$.
The proof of the following theorem uses this connection by constructing formulas that are satisfied in a submodel if and only if the corresponding assignment evaluates to $1$, giving rise to a nice reduction from $\SAT$ to $\ExistSubmodel{\AG}$.

\begin{theorem}\label{th:CTL_A_DelNP-hard}
    $\CTLsubmodel{\AG}$ is $\DelNP$-complete.
\end{theorem}
\begin{proof}
    The upper bound follows from \Cref{th:CTL_InDelNP}.
    By \Cref{th:exist_NP-hard_E_DelNP-hard}, showing $\NPtime$-hardness of $\ExistSubmodel{\AG}$ implies $\DelNP$-hardness of $\CTLsubmodel{\AG}$.

    Let $\varphi$ be propositional formula in negation normal form.
    $\varphi_{\AG}$ is constructed by substituting $x_i$ with $\AG (x_i \rightarrow x_i^1)$ and $\neg x_i$ with $\AG (x_i \rightarrow x_i^0)$ in $\varphi$ for all $x_i \in \Prop(\varphi)$.
    Note that while we have not formally introduced implication, it can simply be taken as  $\lnot x_i \lor x_i^1$.
    Also recall that atomic negation can always be simulated by introducing new propositions and labeling the model accordingly.

    We now show that $\langle \varphi \rangle \in \SAT$, if and only if we have that $\langle \M(\varphi), \varphi_{\AG} \rangle \in \ExistSubmodel{\AF}$.

    Suppose $\varphi \in \SAT$.
    Then there exists an assignment $\mathfrak I$ such that $\mathfrak I(\varphi) = 1$.
    Using $\mathfrak I$, we construct a submodel $\M' = (W', R', \eta, w_0)$ as follows:
    \begin{align*}
        W' & \coloneqq W \setminus \{w_i^k \mid 1 \leq i \leq n, k = 1-\mathfrak I(x_i)\} \\
        R' & \coloneqq R \cap (W' \times W')
    \end{align*}
    That is, we remove the worlds $w_i^1$, if $\mathfrak I(x_i) = 0$ and $w_i^0$, if $\mathfrak I(x_i) = 1$.

    Observe that $\M' \models (x_i \rightarrow x_i^1)$, if and only if $\mathfrak I(x_i) = 1$, since all worlds of $\M'$ labeled with $x_i$ are also labeled with $x_i^1$.
    Analogously, $\M' \models (x_i \rightarrow x_i^0)$, if and only if $\mathfrak I(x_i) = 0$.
    Because $\varphi_{\AG}$ differs from $\varphi$ only in its atoms, it follows that $\M' \models \varphi_{\AG}$ must be true.

    In the same way, if there is a submodel $\M'$ such that $\M' \models \varphi_{\AG}$, we can construct an assignment $\mathfrak I$ from a path $\pi \in \Pi(\M')$ such that $\mathfrak I(\varphi)$ evaluates to $1$.

    To conclude the reduction, observe that the construction of $\M(\varphi)$ and $\varphi_{\AG}$ can both clearly be done in polynomial time, showing $\SAT \pmRed \ExistSubmodel{\AG}$ and proving $\DelNP$-hardness of $\CTLsubmodel{\AG}$.
\end{proof}
Notice that the reduction requires $\AG$ as operator and only the binary Boolean connectors $\land$, $\lor$ and atomic negations, which can be removed by a simple relabeling.

\paragraph*{Fragment AF.}
We will show hardness via relating submodels to deciding the problem HAMPATH~\cite{DBLP:conf/coco/Karp72}.

\begin{definition}\label{def:M(H)}
    Let $H = \langle G, s, t \rangle$ be a $\HAMPATH$ instance, with $G = (V,E)$ a graph and $s, t \in V$.
    We define the Kripke model $\M(H) \coloneqq (W, R, \eta, w_s)$ as follows
    \begin{align*}
        W \coloneqq         & \ \{w_v|v \in V\} \cup \{\hat w\}                                                \\
        R \coloneqq         & \ \{(w_u, w_v)|(u,v) \in E, u \neq t\}\!\cup\! \{(w_t,\hat w), (\hat w,\hat w)\} \\
        \eta(w_v) \coloneqq & \ \{x_v\}, \text{ for } v \in V
    \end{align*}
\end{definition}
The underlying graph of this model is almost $G$ itself, except that a new world $\hat w$ is added, which became the only successor of $w_t$ and has only one relation to itself.
\Cref{fig:ex:M(H)}~(b) depicts such a model for the graph in \Cref{fig:ex:graphG}~(a).

\begin{figure}
    \centering
    \centering
    \begin{tikzpicture}[scale=1.5,y=1cm]
        \node at (-.5,.25) {(a)};
        \node[world] at (0,.75) (s) {$s$};
        \node[world] at (1,.75) (a) {$a$};
        \node[world] at (0,-.25) (b) {$b$};
        \node[world] at (1,-.25) (t) {$t$};
        \draw[-Stealth] (s) -- (a);
        \draw[-Stealth] (a) to[out = 215, in = 55] (b);
        \draw[-Stealth] (b) to[out = 35, in = 235] (a);
        \draw[-Stealth] (b) -- (t);
        \draw[-Stealth] (s) -- (b);
        \draw[-Stealth] (s) -- (t);
        \draw[-Stealth] (t) -- (a);
    \end{tikzpicture}
    \qquad
    \begin{tikzpicture}[scale=1.5,y=1cm]
        \node at (-.5,.5) {(b)};
        \node[world,label={140:$w_s$}] at (0,1) (s) {$x_s$};
        \node[world,label={40:$w_a$}] at (1,1) (a) {$x_a$};
        \node[world,label={180:$w_b$}] at (0,0) (b) {$x_b$};
        \node[world,label={0:$w_t$}] at (1,0) (t) {$x_t$};
        \node[world,label={0:$\hat w$}] at (1.5,0.5) (end) {};

        \draw[Stealth-] (s.west) -- ++(-0.20,0);
        \draw[-Stealth] (s) -- (a);
        \draw[-Stealth] (a) to[out = 215, in = 55] (b);
        \draw[-Stealth] (b) to[out = 35, in = 235] (a);
        \draw[-Stealth] (b) -- (t);
        \draw[-Stealth] (s) -- (b);
        \draw[-Stealth] (s) -- (t);
        \draw[-Stealth] (t) -- (end);
        \draw[-Stealth] (end) edge[loop above, looseness=8] (end);
    \end{tikzpicture}

    \caption{(a) $G$ with Hamiltonian path $s,a,b,t$. (b) Kripke model $\M(H)$ of $H = \langle G, s, t \rangle$.}\label{fig:ex:graphG}\label{fig:ex:M(H)}
\end{figure}
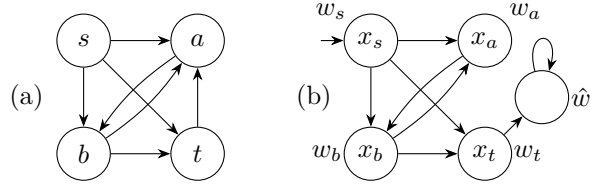

\begin{theorem}\label{th:CTL_AF_land}
    $\CTLsubmodel{\AF}$ is $\DelNP$-complete.
\end{theorem}
\begin{proof}
    The upper bound follows directly from \Cref{th:CTL_InDelNP}.

    Let $H = \langle G, s, t \rangle$ be an instance of $\HAMPATH$ with $G = (V, E), s, t\in V$ and $n = |V|$.
    Further let $\M(H)$ be the Kripke model obtained from $H$ as described in \Cref{def:M(H)}.
    Now,  construct the formula
    $
        \varphi \coloneqq \bigwedge_{v\in V} \AF x_v.
    $

    Notice that a submodel $\M' \subseteq \M$ satisfies $\varphi$ only if it is acyclic.
    That is because all paths have to contain a world labeled with $t$, which only holds at $w_t$.
    The world $w_t$ has a single outgoing edge to $\hat w$, where all paths ``end'' in an infinite loop, making other cycles impossible.

    Next, we have that paths must contain worlds where $v$ for $v \in V$ holds.
    This can only be achieved if all path contain the worlds $w_v$ for $v \in V$.
    It follows that satisfying submodels must contain each world at least once, but because of acyclicity they can also only contain each world at most once.
    Thus satisfying submodels must be single path from $s$ to $t$ containing each world exactly once, i.e., they must be Hamiltonian.
    Construction of $\M(H)$ and $\varphi$ is in $\Ptime$.
\end{proof}
Notice that the reduction requires $\AF$ as operator and the Boolean connector $\land$.

\paragraph*{Fragment AX.}
We again use $\HAMPATH$ to show hardness.
By concatenating the $\AX$ operator $n-1$ times followed by $x_t$, we enforce that submodels must satisfy $x_t$ on all path at position $n$.
Considering the construction of $\M(H)$ this is only possible, if paths are acyclic and contain $w_t$ only at position $n$.
This implies that all satisfying submodels describe Hamiltonian paths from $s$ to $t$.

\begin{theorem}\label{th:DelNP:AX}
    $\CTLsubmodel{\AX}$ is $\DelNP$-complete.
\end{theorem}
\begin{proof}
    The upper bound follows from \Cref{th:CTL_InDelNP}.

    Suppose $H$ is an instance of $\HAMPATH$ and $n = |V|$.
    Then let $\M(H)$ be the Kripke model as defined in \Cref{def:M(H)} and let
    $
        \varphi \coloneqq {\AX}^{n-1}~x_t
    $
    be a formula, where ${\AX}^{n-1}$ denotes the $n-1$-times concatenation of the $\AX$ operator.
    Furthermore, let $\M' \subseteq \M(H)$ be a satisfying submodel.

    First, show that $\pi[n] = w_t$ for all paths $\pi \in \Pi(\M')$.
    \begin{align*}
         & \M', w_s \models {\AX}^{n-1}~x_t                                                                                \\
         & \Leftrightarrow \forall \pi \in \Pi(w_s) : \M', \pi[2] \models {\AX}^{n-2}~x_t
        \tag{Def.}                   \\
         & \Leftrightarrow \forall \pi \in \Pi(w_s) \forall \sigma \in \Pi(\pi[2]): \M', \sigma[2] \models {\AX}^{n-3}~x_t \\
         & \Leftrightarrow \forall \pi \in \Pi(w_s) : \M', \pi[3] \models {\AX}^{n-3}~x_t
        \tag{prefix}                 \\
         & \Leftrightarrow \forall \pi \in \Pi(w_s) : \M', \pi[n-1] \models \AX x_t
        \tag{repeat}                 \\
         & \Leftrightarrow \forall \pi \in \Pi(w_s) : \M', \pi[n] \models x_t
    \end{align*}
    By the definition of $\M(H)$, only $\eta(w_t) = x_t$.
    Thus $\forall \pi \in \Pi(w_s)$ we have $\pi[n] = w_t$.

    Note that $w_t$ cannot be on any path before that.
    Otherwise the path could only continue to $\hat w$ and ``end'' there.
    Also, submodels again cannot have cycles, otherwise there would be a path that never reaches $w_t$.
    So we can conclude that on all paths in $\M'$ the first $n$ elements must be different.
    With $n$ worlds other than $\hat w$, this leads to satisfying submodels that are Hamiltonian paths from $w_s$ to $w_t$, showing correctness of the reduction.
    The reduction can be computed in $\Ptime$.
\end{proof}
Notice that the reduction merely requires $\AX$ as operator and no Boolean connectors are used.

\begin{figure}
    \centering
        \begin{tikzpicture}
            [
                scale=1,y=1cm
            ]
            \node[world,label={170:$w_s$},submodel] at (-2,3.5) (s) {};
            \node[world,smallWorld] at (-3.10,2.25) (s1) {$x_1$};
            \node[world,smallWorld] at (-2.35,2.25) (s2) {$x_2$};
            \node[world,smallWorld] at (-1.65,2.25) (s3) {$x_3$};
            \node[world,smallWorld,submodel] at (-0.90,2.25) (s4) {$x_4$};
            \node[world,label={170:$\hat{w}_s$},submodel] at (-2,1) (s-h) {};

            \begin{scope}[xshift=-.9cm]
                \node[world,label={ 10:$w_a$},submodel] at (2,3.5) (a) {};
                \node[world,smallWorld] at (0.90,2.25) (a1) {$x_1$};
                \node[world,smallWorld] at (1.65,2.25) (a2) {$x_2$};
                \node[world,smallWorld,submodel] at (2.35,2.25) (a3) {$x_3$};
                \node[world,smallWorld] at (3.10,2.25) (a4) {$x_4$};
                \node[world,label={ 10:$\hat{w}_a$},submodel] at (2,1) (a-h) {};
            \end{scope}

            \begin{scope}[yshift=.6cm]
                \node[world,label={170:$w_b$},submodel] at (-2,-1) (b) {};
                \node[world,smallWorld] at (-3.10,-2.25) (b1) {$x_1$};
                \node[world,smallWorld,submodel] at (-2.35,-2.25) (b2) {$x_2$};
                \node[world,smallWorld] at (-1.65,-2.25) (b3) {$x_3$};
                \node[world,smallWorld] at (-0.90,-2.25) (b4) {$x_4$};
                \node[world,label={170:$\hat{w}_b$},submodel] at (-2,-3.5) (b-h) {};

                \begin{scope}[xshift=-.9cm]
                    \node[world,label={ 10:$w_t$},submodel] at (2,-1) (t) {};
                    \node[world,smallWorld,submodel] at (0.90,-2.25) (t1) {$x_1$};
                    \node[world,smallWorld] at (1.65,-2.25) (t2) {$x_2$};
                    \node[world,smallWorld] at (2.35,-2.25) (t3) {$x_3$};
                    \node[world,smallWorld] at (3.10,-2.25) (t4) {$x_4$};
                    \node[world,label={ 10:$\hat{w}_t$},submodel] at (2,-3.5) (t-h) {$x_t$};
                \end{scope}
            \end{scope}

            \draw[-Stealth] (s) -- (s1);
            \draw[-Stealth] (s1) -- (s-h);
            \draw[-Stealth] (s) -- (s2);
            \draw[-Stealth] (s2) -- (s-h);
            \draw[-Stealth] (s) -- (s3);
            \draw[-Stealth] (s3) -- (s-h);
            \draw[-Stealth,submodel] (s) -- (s4);
            \draw[-Stealth,submodel] (s4) -- (s-h);

            \draw[-Stealth] (a) -- (a1);
            \draw[-Stealth] (a1) -- (a-h);
            \draw[-Stealth] (a) -- (a2);
            \draw[-Stealth] (a2) -- (a-h);
            \draw[-Stealth,submodel] (a) -- (a3);
            \draw[-Stealth,submodel] (a3) -- (a-h);
            \draw[-Stealth] (a) -- (a4);
            \draw[-Stealth] (a4) -- (a-h);

            \draw[-Stealth] (b) -- (b1);
            \draw[-Stealth] (b1) -- (b-h);
            \draw[-Stealth,submodel] (b) -- (b2);
            \draw[-Stealth,submodel] (b2) -- (b-h);
            \draw[-Stealth] (b) -- (b3);
            \draw[-Stealth] (b3) -- (b-h);
            \draw[-Stealth] (b) -- (b4);
            \draw[-Stealth] (b4) -- (b-h);

            \draw[-Stealth,submodel] (t) -- (t1);
            \draw[-Stealth,submodel] (t1) -- (t-h);
            \draw[-Stealth] (t) -- (t2);
            \draw[-Stealth] (t2) -- (t-h);
            \draw[-Stealth] (t) -- (t3);
            \draw[-Stealth] (t3) -- (t-h);
            \draw[-Stealth] (t) -- (t4);
            \draw[-Stealth] (t4) -- (t-h);

            \draw[Stealth-,submodel] (s.west) -- ++(-0.5,0);
            \draw[-Stealth] (s-h) -- (b);
            \draw[-Stealth,submodel] (s-h) to[out = 0, in = 180] (a);
            \draw[-Stealth] (s-h) -- (t);
            \draw[-Stealth,submodel] (a-h) -- (b);
            \draw[-Stealth] (b-h) to[out = 0, in = 180, looseness=0.9] (a);
            \draw[-Stealth,submodel] (b-h) to[out = 0, in = 180] (t);
            \draw[-Stealth,submodel] (t-h) edge[loop left, >=Stealth] (t-h);
        \end{tikzpicture}
        \quad
        \begin{tikzpicture}
            [
                scale=1,
            ]
            \begin{scope}[yshift=-0.6cm]
                \node[world,label={160:$w_s$}, align=center, submodel] at (-2,5) (s) {$z,x_1,$\\$x_2,x_3,x_4$};
                \node[world,smallWorld] at (-3.10,3.75) (s1) {$x_1$};
                \node[world,smallWorld] at (-2.35,3.75) (s2) {$x_2$};
                \node[world,smallWorld] at (-1.65,3.75) (s3) {$x_3$};
                \node[world,smallWorld, submodel] at (-0.90,3.75) (s4) {$x_4$};
                \node[world,label={160:$\tilde{w}_s$}, align=center, submodel] at (-2,2.5) (s-t) {$y,x_1,$\\$x_2,x_3,x_4$};
                \node[world,label={160:$\hat{w}_s$}, submodel] at (-2,1) (s-h) {$y, z$};

                \begin{scope}[xshift=-.9cm]
                    \node[world,label={ 20:$w_a$}, align=center, submodel] at (2,5) (a) {$z,x_1,$\\$x_2,x_3,x_4$};
                    \node[world,smallWorld] at (0.90,3.75) (a1) {$x_1$};
                    \node[world,smallWorld] at (1.65,3.75) (a2) {$x_2$};
                    \node[world,smallWorld, submodel] at (2.35,3.75) (a3) {$x_3$};
                    \node[world,smallWorld] at (3.10,3.75) (a4) {$x_4$};
                    \node[world,label={ 20:$\tilde{w}_a$}, align=center, submodel] at (2,2.5) (a-t) {$y,x_1,$\\$x_2,x_3,x_4$};
                    \node[world,label={ 20:$\hat{w}_a$}, submodel] at (2,1) (a-h) {$y, z$};
                \end{scope}
            \end{scope}

            \node[world,label={160:$w_b$}, align=center, submodel] at (-2,-1) (b) {$z,x_1,$\\$x_2,x_3,x_4$};
            \node[world,smallWorld] at (-3.10,-2.25) (b1) {$x_1$};
            \node[world,smallWorld, submodel] at (-2.35,-2.25) (b2) {$x_2$};
            \node[world,smallWorld] at (-1.65,-2.25) (b3) {$x_3$};
            \node[world,smallWorld] at (-0.90,-2.25) (b4) {$x_4$};
            \node[world,label={160:$\tilde{w}_b$}, align=center, submodel] at (-2,-3.5) (b-t) {$y,x_1,$\\$x_2,x_3,x_4$};
            \node[world,label={160:$\hat{w}_b$}, submodel] at (-2,-5) (b-h) {$y, z$};

            \begin{scope}[xshift=-.9cm]
                \node[world,label={ 20:$w_t$}, align=center, submodel] at (2,-1) (t) {$z,x_1,$\\$x_2,x_3,x_4$};
                \node[world,smallWorld, submodel] at (0.90,-2.25) (t1) {$x_1$};
                \node[world,smallWorld] at (1.65,-2.25) (t2) {$x_2$};
                \node[world,smallWorld] at (2.35,-2.25) (t3) {$x_3$};
                \node[world,smallWorld] at (3.10,-2.25) (t4) {$x_4$};
                \node[world,label={ 20:$\tilde{w}_t$}, align=center, submodel] at (2,-3.5) (t-t) {$y,x_1,$\\$x_2,x_3,x_4$};
                \node[world,label={ 20:$\hat{w}_t$}, submodel] at (2,-5) (t-h) {$y, z$};
            \end{scope}
            \foreach \node/\b in {s/4,a/3,b/2,t/1}{
                    \draw[-Stealth, submodel] (\node-t) -- (\node-h);
                    \foreach \num in {1,2,3,4}{
                            \ifnum \num=\b
                                \draw[-Stealth, submodel] (\node) -- (\node\num);
                                \draw[-Stealth, submodel] (\node\num) -- (\node-t);
                            \else
                                \draw[-Stealth] (\node) -- (\node\num);
                                \draw[-Stealth] (\node\num) -- (\node-t);
                            \fi
                        }
                }

            \draw[Stealth-, submodel] (s.west) -- ++(-0.5,0);
            \draw[-Stealth] (s-h) -- (b);
            \draw[-Stealth, submodel] (s-h) to[out = 0, in = 208] (a);
            \draw[-Stealth] (s-h) -- (t);
            \draw[-Stealth, submodel] (a-h) -- (b);
            \draw[-Stealth] (b-h) to[out = 0, in = 208, looseness=0.66] (a);
            \draw[-Stealth, submodel] (b-h) to[out = 0, in = 208] (t);
            \draw[-Stealth, submodel] (t-h) edge[loop left,>=Stealth] (t-h);
        \end{tikzpicture}
    \caption{Kripke model (left) of $\M_{\AU}(H)$ and (right) of $\M_{\AR}(H)$ for the graph in \Cref{fig:ex:graphG} (a). The highlighted worlds and relations form a submodel that induces a Hamiltonian path for the instance $H = \langle G, s, t \rangle$ and satisfy (left) $\varphi_4 = ((((\top \AU x_t) \AU x_1) \AU x_2) \AU x_3) \AU x_4$ and (right) $\varphi_4 = (\cdots(\top \AR z) \AR y) \AR x_1)\AR z) \AR y) \AR x_2) \AR z)\AR y)\AR x_3)\AR z) \AR y)\AR x_4$.}
    \label{fig:ex:M_AU(H)}    \label{fig:ex:M_AR(H)}
\end{figure}
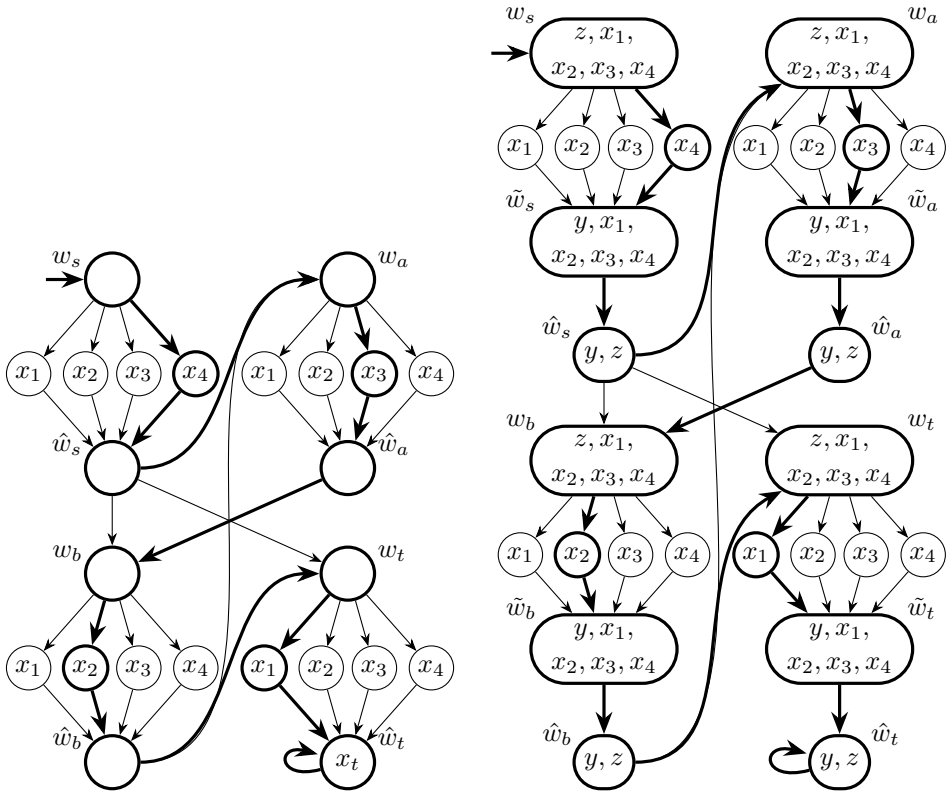

\paragraph*{Fragment AU.}
We continue to use $\HAMPATH$.
For fragment $\AU$, we construct a new submodel $\M_{\AU}(H)$ that expands each node into a ``diamond'' construct with a world for the incoming relations and a world for the outgoing relations, as well as a number of intermediate worlds equal to the total number of vertices in $G$.
We then construct an $\AU$ formula such that all paths in a satisfying submodel are acyclic and contain a different intermediate world at each ``diamond'', thereby describing a Hamiltonian path of $G$.
\begin{definition}\label{def:M_AU(H)}
    Let $G = (V,E)$ be a graph, $s, t \in V$, $n = |V|$, and $H= \langle G, s, t \rangle$ be an instance of $\HAMPATH$.
    We define the model $\M_{\AU}(H) \coloneqq (W, R, \eta, w_s)$ as follows:
    \begin{align*}
        W \coloneqq             & \ \{w_v,\hat{w}_v, w_{v,i} \mid v \in V, 1 \leq i \leq n\}                        \\
        R \coloneqq             & \ \{(w_v,w_{v,i}), (w_{v,i}, \hat{w}_v) \mid v \in V, 1 \leq i \leq n\}           \\
                                & \cup\{(\hat{w}_u, w_v) \mid (u,v) \in E, u \neq t\}\cup \{(\hat{w}_t,\hat{w}_t)\} \\
        \eta(w_{v,i}) \coloneqq & \ \{x_i\} \text{ for } 1 \leq i \leq n,\quad\eta(\hat{w}_t) \coloneqq \{x_t\}
    \end{align*}
\end{definition}
\Cref{fig:ex:M_AU(H)} depicts the submodel $\M_{\AU}(H)$ constructed from the graph in \Cref{fig:ex:graphG} (a), with $H = \langle G, s, t \rangle$.

\begin{theorem}\label{th:DelNP:AU}
    $\CTLsubmodel{\AU}$ is $\DelNP$-complete.
\end{theorem}
\begin{proof}
    The upper bound follows directly from \Cref{th:CTL_InDelNP}.

    Let $H = \langle G, s, t \rangle$ be an instance of $\HAMPATH$ and $\varphi_n$ be the formula of interest here, for
    $$
        \varphi_i \coloneqq \begin{cases}
            \varphi_{i-1} \AU x_{i} & \text{if } i > 0, \\
            \top \AU x_t            & \text{if } i = 0.
        \end{cases}
    $$
    We show that any submodel $\M_{\AU}'(H) \subseteq \M_{\AU}(H)$ that satisfies $\varphi$ must consist of a single infinite path, which contains all $w_v$ for $v \in V$ exactly once, starts in $w_s$ and ends with an infinite sequence of $\hat{w}_t$.
    Thereby showing $H \in \HAMPATH$.

    Since the root remains unchanged by our definition of submodels, all paths in any $\M_{\AU}'$ begin with $w_s$.
    While it should be obvious that if a path contains $w_t$ it ``ends'' there.
    Contrary it might not be immediately clear, that by our definition of $\varphi_n$ all path must contain $w_t$ at some point.
    \begin{align*}
         & \M_{\AU}'(H), w_s \models \varphi_n ( = \varphi_{n-1} \AU x_n)                                                   \\
         & \Rightarrow \forall \pi \in \Pi(w_s) \exists k \geq 1 \forall i < k : \M'_{\AU}(H), \pi[i] \models \varphi_{n-1}
        \tag{Right hand side of \AU}  \\
         & \Rightarrow \forall \pi \in \Pi(w_s) : \M'_{\AU}(H), \pi[1] \models \varphi_{n-2} \AU x_{n-1}
        \tag{take $i = 1$}            \\
         & \Rightarrow \M'_{\AU}(H), w_s \models \varphi_{n-2} \AU x_{n-1}
        \tag{$\pi[1] = w_s$}          \\
         & \makebox[\widthof{$\ \Rightarrow$}][c]{$\vdots$} \tag{repeat}                                                    \\
         & \Rightarrow \M'_{\AU}(H), w_s \models \varphi_0                                                                  \\
         & \Rightarrow \M'_{\AU}(H), w_s \models \top \AU x_t                                                               \\
         & \Rightarrow \M'_{\AU}(H), w_s \models \AF x_t
        \tag{\Cref{ob:CTL_equis}}
    \end{align*}
    From this we conclude that to satisfy $\AF x_t$ all paths in $\M_{\AU}'$ contain a world labeled $x_t$ at some point.
    By the construction of $\M_{AU}$ this has to be $\hat{w}_t$.

    A consequence of all paths leading to $\hat{w}_t$ is that the underlying graph of $\M'_{\AU}(H)$ must be acyclic, except for the loop at the world $\hat{w}_t$.
    Assume $\M'_{\AU}(H)$ has a cycle.
    The cycle obviously cannot contain $\hat{w}_t$, because $\hat{w}_t$ has no outgoing relations except to itself, which contradicts the above statement.
    \begin{align*}
         & \M'_{\AU}(H), w_s \models \varphi_n ( = \varphi_{n-1} \AU x_n)                                         \\
         & \Rightarrow \forall \pi \in \Pi(w_s) \exists k \geq 1 : \M'_{\AU}(H), \pi[k] \models x_n \text{ and } \\
         & \qquad \forall i < k: \M'_{\AU}(H), \pi[i] \models \varphi_{n-1}
        \tag{definition AU} \\
         & \Rightarrow \forall \pi \in \Pi(w_s) \exists k \geq 1 : \M'_{\AU}(H), \pi[k] \models x_n \text{ and } \\
         & \qquad \M'_{\AU}(H), \pi[k-1] \models \varphi_{n-2} \AU x_{n-1}
        \tag{take $i = k-1$} \\
         & \Rightarrow \forall \pi \in \Pi(w_s) \exists k \geq 1 : \M'_{\AU}(H), \pi[k] \models x_n \text{ and }  \\
         & \qquad \forall \pi' \in \Pi(k-1) \exists j \geq 1 : \M'_{\AU}(H), \pi'[j] \models x_{n-1} \text{ and } \\
         & \qquad \forall i < j: \M'_{\AU}(H), \pi'[i] \models \varphi_{n-2}
        \tag{definition \AU}\\
         & \Rightarrow \forall \pi \in \Pi(w_s) \exists k \geq 1 : \M'_{\AU}(H), \pi[k] \models x_n \text{ and }  \\
         & \qquad \exists k' \geq k-1 :  \M'_{\AU}(H), \pi[k'] \models x_{n-1} \text{ and }                       \\
         & \qquad \forall i < j: \M'_{\AU}(H), \pi[i] \models \varphi_{n-2}
        \tag{all $\pi'$ in $\Pi(R)$ with some prefix}                                                             \\
         & \Rightarrow \forall \pi \in \Pi(w_s) \exists k \geq 1 : \M'_{\AU}(H), \pi[k] \models x_n \text{ and }  \\
         & \qquad \exists k' \geq k-1 :  \M'_{\AU}(H), \pi[k'] \models x_{n-1} \text{ and }                       \\
         & \qquad \M'_{\AU}(H), \pi[k'-1] \models \varphi_{n-2}
        \tag{take $i = k'-1$}
    \end{align*}
    Repeating these steps leads to the following
    \begin{align*}
         & \forall \pi \in \Pi(w_s) \exists k \geq 1 : \M'_{\AU}(H), \pi[k] \models x_n                 \\
         & \text{ and } \exists k' \geq k-1 :  \M'_{\AU}(H), \pi[k'] \models x_{n-1}                    \\
         & \text{ and } \dots                                                                           \\
         & \text{ and }\exists k^{(n-1)} \geq k^{(n-2)} - 1 : \M'_{\AU}(H), \pi[k^{(n-1)}] \models x_1.
    \end{align*}
    Notice that in the construction of $\M_{\AU}(H)$ the predecessor of worlds labeled with $w_i$ has no label themselves.
    Also notice that the worlds labeled with $w_i$ have no other label $w_j$ for $j \neq i$.
    Therefore $k^{(i)} > k^{(i-1)}$ instead of $k^{(i)} \geq k^{(i-1)}-1$ must hold for $0 \leq i \leq n-1$.

    From this we can conclude that all paths of $\M'_{\AU}(H)$ have to contain worlds labeled with $x_n$ to $x_1$.
    Since labeled worlds can only be reached from worlds $w_v$ for $v \in V$, each $w_v$ can be on a path only once due to the acyclicity of $\M'_{\AU}(H)$.
    There are $n$ worlds $w_v$ in total, which means that all $w_v$ must be on all paths of $\M'_{\AU}(H)$ exactly once.
    Obviously this can only be the case if there is only one path in $\Pi(\M'_{\AU}(H))$
    because second or more paths would either introduce a cylce or would not contain all $w_v$.

    By this we can conclude that 
    $$
    H \in \HAMPATH \iff \langle \M_{\AU}', \varphi_n \rangle \in \ExistSubmodel{\AU}.
    $$
    Notice that $\M_{\AU}(H)$ and $\varphi$ can be computed in polynomial time.
    Therefore $\HAMPATH \pmRed \ExistSubmodel{\AU}$ and with \Cref{th:exist_NP-hard_E_DelNP-hard} it follows that $\CTLsubmodel{\AU}$ is $\DelNP$-complete.

\end{proof}

Notice that the reduction merely requires $\AU$ as operator and no Boolean connectors are used.

\paragraph*{Fragment AR.}
We use a similar ``diamond'' expansion of the nodes in $G$, as in the proof for $\AU$.
Here an extra world is added to the construct between the middle worlds and the world for outgoing relations.
In addition, the labeling is extended to make $\AR$ behave in the indented way.
That is, we want to repeatedly force the left hand side of the $\AR$ operator in our constructed formula $\varphi$ to only hold in specific subsequent worlds to simulate the behavior of the $\AX$ operator.
The construction of $\varphi$, also requires that worlds labeled with $x_1, x_2, \dots, x_n$ are on the paths, similar to the proof of the $\AU$ fragment.
\begin{definition}\label{def:M_AR(H)}
    Let $G = (V,E)$ be a graph with $n = |V|$ and $H \coloneqq \langle G, s, t \rangle$ an instance of $\HAMPATH$.
    Define $\M_{\AR}(H) \coloneqq (W, R, \eta, w_s)$ as follows (see \Cref{fig:ex:M_AR(H)} for an example):
    \begin{align*}
        W \coloneqq             & \ \{w_v,\tilde{w}_v,\hat{w}_v, w_{v,i} \mid v \in V, 1 \leq i \leq n\}                 \\
        R \coloneqq             & \ \left\{
        \begin{array}{@{}l}
            (w_v,w_{v,i}),          \\
            (w_{v,i}, \tilde{w}_v), \\
            (\tilde{w}_v, \hat{w}_v)
        \end{array}
        \middle\vert\, v \in V, 1 \leq i \leq n \right\}                                                                 \\
                                & \cup\{(\hat{w_u}, w_v) \mid (u,v) \in E \text{ and } u \neq t\}                        \\
                                & \cup \{(\hat{w}_t,\hat{w}_t)\}                                                         \\
        \eta(w_{v,i}) \coloneqq & \ \{x_i\} \text{ for } 1 \leq i \leq n,\quad         \eta(\hat{w}_v) \coloneqq \{z,y\} \\
        \eta(w_v) \coloneqq     & \ \{z, x_1, \dots, x_n\}, \eta(\tilde{w}_v) \coloneqq \{y, x_1, \dots, x_n\}
    \end{align*}
\end{definition}
\begin{theorem}\label{th:DelNP:AR}
    $\CTLsubmodel{\AR}$ is $\DelNP$-complete.
\end{theorem}
\begin{proof}
    The upper bound follows from \Cref{th:CTL_InDelNP} again.
    For the lower bound we will again reduce from $\HAMPATH$, using the Kripke model $\M_{\AR}(H)$ defined in \Cref{def:M_AR(H)} and $\varphi_n$ with
    $$
        \varphi_i \coloneqq \begin{cases}
            ((\varphi_{i-1} \AR z) \AR y) \AR x_i & \text{if } i > 0, \\
            \top                                  & \text{if } i = 0.
        \end{cases}
    $$

    We will show that $\langle G, s, t \rangle \mapsto \langle \M_{\AR}(H), \varphi_n \rangle$ is a reduction function for $\HAMPATH \pmRed \ExistSubmodel{\AR}$.

    Suppose $\langle G, s, t \rangle$ is an instance of $\HAMPATH$.
    Further let $\M'_{\AR}$ be a satisfying submodel of $\M_{\AR}(H)$.
    The following claim shows that $\varphi_n$ forces $\M'_{\AR}$ to only contain paths of a certain form.
    \begin{claim}\label{clm:ARforcesSequence}
        For any path $\pi \in \Pi(\M'_{\AR})$ we have that for all $0 \leq i < n$ exists a $v \in V$ such that
        \begin{align*}
            \pi[4i+1] & = w_v,         & \pi[4i+2] & = w_{v,n - i}, \\
            \pi[4i+3] & = \tilde{w}_v, & \pi[4i+4] & = \hat{w}_v
        \end{align*}
        and
        \begin{align*}
            \M'_{\AR}, \pi[4i+1] & \models ((\varphi_{n-i-1} \AR z) \AR y) \AR x_{n-i} \\
            \M'_{\AR}, \pi[4i+3] & \models (\varphi_{n-i-1} \AR z) \AR y               \\
            \M'_{\AR}, \pi[4i+4] & \models \varphi_{n-i-1} \AR z.
        \end{align*}
    \end{claim}
    \begin{proof}
        We proceed by induction on $i$.
        For the base case $i = 0$, it is clear that $\pi[1] = w_s$, since all paths start at the root.
        Notice that from
        \[
            \M'_{\AR}, w_s \models ((\varphi_{n-1} \AR z) \AR y) \AR x_n
        \]
        and
        \[
            \M'_{\AR}, w_s \not\models (\varphi_{n-1} \AR z) \AR y
        \]
        it follows that $\M'_{\AR}, \pi[2] \models x_n$.
        By the construction of $\M_{\AR}(H)$ the only world directly after $w_s$ labeled with $x_n$ is $w_{s,n}$, thus we have $\pi[2] = w_{s,n}$.
        It is then obvious that $\pi[3] = \tilde{w}_s$ and $\pi[4] = \hat{w}_s$,
        because $\tilde{w}_s$ is the only successor of $w_{s,n}$ and has only $\hat{w}_s$ as its successor.

        Now consider how this unravels $\varphi_n$.
        First observe that $\M'_{\AR}, \pi[4] \not\models x_n$, it follows that there exists an $i < 4$ such that
        \[
            \M'_{\AR}, \pi[i] \models (\varphi_{n-1} \AR z) \AR y.
        \]
        We have already established that $i$ cannot be $1$.
        It also cannot be $2$, because $y \not\in \eta(\pi[2]) (= \eta(w_{s,n}))$
        Therefore, if $\M'_{\AR}$ is a satisfying submodel, then $i = 3$ and
        \[
            \M'_{\AR}, \pi[3] \models (\varphi_{n-1} \AR z) \AR y.
        \]
        Similarly, since all successors of $\hat{w}_s$ in $\M_{\AR}(H)$ (and thus in all its submodels) are not labeled with $y$ and $\tilde{w}_s$ is not labeled with $z$, it follows that
        \[
            \M'_{\AR}, \pi[4] \models \varphi_{n-1} \AR z.
        \]

        \medskip

        For the induction step, we have
        \[
            \M'_{\AR}, \pi[4i+4] \models \varphi_{n-i-1} \AR z
        \]
        and $\pi[4i+4] = \hat{w}_v$ from the induction hypothesis.
        Notice that all successors $w_u$ of $\hat{w}_v$ are labeled with $z$, while their successors are not.
        Therefore
        \begin{align*}
            \M'_{\AR}, \pi[4i+5] \models ((\varphi_{n-i-2} \AR z) \AR y) \AR x_{n-i-1}
        \end{align*}
        with $\pi[4i+5] = w_u$ for some $u \in V$.
        Note again that this formula can only hold at $\pi[4i+5]$ and not $\pi[4i+4]$, because $x_{n-i-1} \not\in \eta(\pi[4i+4])$.
        Similar to the base case, the only successor of $w_u$ labeled with $x_{n-i-1}$ is $w_{u,n-i-1}$, therefore $\pi[4i+6] = w_{u,n-i-1}$.
        $\pi[4i+7] = \tilde{w}_u$ and $\pi[4i+8] = \hat{w}_u$ follow immediately.

        We can again observe the unraveling of the formula.
        \begin{equation*}
            \M'_{\AR}, \pi[4i+3] \models (\varphi_{n-i-2} \AR z) \AR y
        \end{equation*}
        must be true because, $y \not\in \eta(\pi[4i+6])$ and $x_{n-i-2} \not\in \eta(\pi[4i+8])$.
        Also
        \begin{equation*}
            \M'_{\AR}, \pi[4i+8] \models \varphi_{n-i-2} \AR z
        \end{equation*}
        because $z \not\in \eta(\pi[4i+7])$ and no successor of $\hat{w}_u$ has label $y$.
    \end{proof}

    It follows from \Cref{clm:ARforcesSequence} that all paths of a submodel satisfying $\varphi$ visit $\geq n$ worlds $w_{v,i}$.
    \Cref{clm:ARforcesSequence} also shows that
    \[
        \M'_{\AR}, w_v \models ((\varphi_{n-i-1} \AR z) \AR y) \AR x_{n-i}
    \]
    which forces all successors of $w_v$ to be labeled with $x_{n-i}$.
    For this to be true $w_{v,n-i}$ has to be the only successor of $w_v$.
    From this we can conclude that all $w_{v,n-i}$ must have different $v$.
    With $\geq n$ worlds $w_{v,i}$ on any path and $|V| = n$, it follows that all paths visit all worlds $w_v$ once.
    Notice that by our construction of the Kripke model, world $\hat{w}_t$ is a ``dead-end'' and therefore must be visited last.

    This shows that all satisfying submodels of $\M_{\AR}(H)$ must describe a Hamiltonian path of $G$.
    The reduction function is computable in polynomial time, since both the model $\M_{\AR}(H)$ and the formula $\varphi$ can be constructed done in polynomial time with respect to the graph $G$.
\end{proof}

Notice that the reduction merely requires $\AR$ as operator and no Boolean connectors are used.

\paragraph*{Fragment EX, EF, EG, EU \& ER.}
$\DelNP$-hardness of existentially quantified operators follows immediately, when considering negation.
\begin{corollary}
    For $\emptyset \neq \mathcal O \subseteq \{\EX, \EF, \EG, \EU, \ER\}$ we have that $\CTLsubmodel{\mathcal O}$ is $\NPtime$-complete.
\end{corollary}
\begin{proof}
    Follows directly form the duality between the existential and universal path quantifiers (see \Cref{ob:CTL_equis}) and our results for the universally quantified cases.
\end{proof}
Notice that for the fragments $\EX$, $\EU$ and $\ER$ negation suffices as the only Boolean connector to archive intractability.
In contrast the fragments $\EF$ and $\EG$ require all Boolean connectors.
We summarise all results in one statement.
\begin{theorem}
    Let $\emptyset \neq \mathcal O\subseteq\ALL$ be a set of $\CTL$ operators.
    Then $\CTLsubmodel{\mathcal O}$ is $\DelNP$-complete.
\end{theorem}

\section{The silver lining}
In this section, we strive for tractability results.
For this we restrict also the allowed Boolean connectors and accordingly require to extend the problem notion a bit as follows.
For instance, we will write $\ExtendSubmodel{\EX, \ER, \EU}$ whenever we restrict the formulas to \emph{only} the operators $\EX, \ER, \EU$ without \emph{any} Boolean connectors.

\paragraph*{Fragment EX, ER, EU \& Conjunction, Disjunction.}
The first tractability result we present is a restriction to formulas only containing existentially quantified $\CTL$ operators and no negation.
That is, we show $\DelP$ membership of $\CTLsubmodel{\EX, \ER, \EU, \land, \lor}$.
Recall that we have $\EF \varphi = \top \EU \varphi$ and $\EG \varphi = \varphi \ER \bot$.

The next Lemma gives a straightforward way to decide $\ExtendSubmodel{\EX, \ER, \EU, \land, \lor}$, by only having to consider the model and partial solution.
\begin{lemma}\label{lem:M_models_if_M'_models}
    Let $\M' \subseteq \M$ be a submodel.
    If $\M\not\models \varphi$, for any $\{\EX, \EU, \ER, \land, \lor\}$-formula $\varphi$, then $\M' \not\models \varphi$.
\end{lemma}

\begin{proof}
    To prove this lemma consider its contraposition, i.e., $\M' \models \varphi$ implies $\M\models \varphi$.
    Note that the set of paths that satisfy $\varphi$ in $\M'$ also exist in $\M$.
    Since $\varphi$ does not contain negation, the same set of paths must satisfy $\varphi$ in $\M$.
\end{proof}

\begin{theorem}
    $\CTLsubmodel{\EX, \ER, \EU, \land, \lor} \in \DelP$
\end{theorem}
\begin{proof}
    We now briefly describe a simple deterministic polynomial time al\-go\-rithm that decides $\ExtendSubmodel{\EX, \ER, \EU, \land, \lor}$.
    By \Cref{lem:M_models_if_M'_models}, if, for a partial solution, we have that ${\M - D} \not\models \varphi$, then it cannot be extended to an actual solution.
    Conversely, if $\M - D \models \varphi$ is true, then the empty extension is sufficient.
    Thus, any polynomial time model checking algorithm on an instance $\langle \M - D, \varphi\rangle$ can be used to decide $\ExtendSubmodel{\EX, \ER, \EU, \land, \lor}$.
\end{proof}

\paragraph*{Fragment AF \& AG.}

We adapted a result presented by Krebs et al.~\cite[Lemma 10]{DBLP:journals/acta/KrebsMM19}, showing that every $\{\AF,\AG\}$-formula can be reduced to contain at most two temporal operators.
\begin{lemma}\label{lem:FG_trimming}
    For any formula $\varphi$ we have that
    \begin{enumerate} 
        \item $\AF\AF \varphi \equiv \AF \varphi$
        \item $\AG\AG \varphi \equiv \AG \varphi$
        \item $\AG\AF\AG \varphi \equiv \AF\AG \varphi$
        \item $\AF\AG\AF \varphi \equiv \AG\AF \varphi$
    \end{enumerate}
\end{lemma}
\begin{proof}
    (1)
    \begin{align*}
         & \M, w \models \AF\AF \varphi                      \\
         & \Leftrightarrow \forall \pi \in \Pi(w) \exists k \geq 1 \forall \sigma \in \Pi(\pi[k]) \exists j \geq 1 : \M, \sigma[j] \models \varphi \\
         & \Leftrightarrow \forall \pi \in \Pi(w) \exists k \geq 1 \exists j \geq k : \M, \pi[j] \models \varphi                                   \\
         & \Leftrightarrow \forall \pi \in \Pi(w) \exists k \geq 1 : \M, \pi[k] \models \varphi                                                    \\
         & \Leftrightarrow \M, w \models \AF \varphi
    \end{align*}

    (2) analogously.

    (3) $\M, w \models \AG\AF\AG \varphi \Rightarrow \M, w \models \AF\AG\varphi$ is trivial.
    For the other direction, assume $\M, w \models \AF\AG\varphi$.
    Let $\pi \in \Pi(w)$ be an arbitrary path.
    $\M, \pi[k] \models \AG\varphi$ holds for some $k$ with $\M, \pi[i] \not\models \AG\varphi$ for all $i < k$.
    With $\M, \pi[1] \models \AF \AG\varphi$, it follows that $\M, \pi[i] \models \AF \AG\varphi$ for all $i < k$.
    Further take some $\sigma \in \Pi(\pi[k])$.
    From $\M, \pi[k] \models \AG\varphi$ it follows that $\M, \sigma[j] \models \AG\varphi$ which leads to $\M, \sigma[j] \models \AF\AG\varphi$ for all $j \geq 1$.
    Therefore, we have an infinite path $\rho = \pi[1], \pi[2], \dots \pi[k-1], \sigma[1] (= \pi[k]), \sigma[2]$ with $\M, \rho[i] \models \AF \AG\varphi$ for all $i \geq 1$.
    Since $\pi$ and $\sigma$ are arbitrary, this holds for all $\rho \in \Pi(w)$, so $\M, w \models \AG\AF\AG\varphi$.

    (4) $\M, w \models \AG\AF\varphi \Rightarrow \M, w \models \AF\AG\AF\varphi$ is trivial.
    For the other direction, assume $\M, w \models \AF\AG\AF\varphi$.
    Now suppose $\M, w \not\models \AG\AF\varphi$.
    By the duality of $\AG$ and $\AF$ it follows that $\M, w \models \EF\EG \lnot\varphi$, but this cannot be true without contradicting our assumption.
    On a path $\pi \in \Pi(w)$ witnessing this there would be a $k \geq 1$ such that for all $i \geq k : \M, \pi[i] \models \lnot\varphi$.
    But this contradicts our assumption that on all path, there would be an $k \geq 1$ such that for all $i \geq k$ there is an $h \geq i : \M, \pi[h] \models\varphi$.
    We can therefore conclude that $\M, w \models \AG\AF\varphi$.
\end{proof}

\begin{theorem}\label{th:FG_inDelP}
    $\CTLsubmodel{\AF,\AG}$ is in $\DelP$.
\end{theorem}
\begin{proof}
    The following algorithm decides $\ExtendSubmodel{\AF,\AG}$ deterministically and in polynomial time.

    The input is $\langle \M, \varphi, D \rangle$, where $\M = (W, R, \eta, r)$ is a Kripke model, $\varphi$ is a $\{\AF,\AG\}$-formula, and $D$ is a set of deletions.
    Let $\M' = (W', R', \eta, r) \coloneqq \M - D$ be current submodel and $\varphi'$ be the shortened formula obtained from $\varphi$ using \Cref{lem:FG_trimming}.
    Notice that $\varphi'$ can only have one of four forms.

    Now, the algorithm has the following behaviour, depending on $\varphi'$, where $x$ is in $\Prop$:
    \begin{itemize}
        \item $\varphi' = \AF x$: if $\M' \models \EF x$ accept, else reject.
        \item $\varphi' = \AG x$: if $\M' \models \EG x$ accept, else reject.
        \item $\varphi' = \AF\AG x$: if $\M' \models \EF\EG x$ accept, else reject.
        \item $\varphi' = \AG\AF x$: let $\hat{\M} = (W', R', \hat{\eta}, r)$ be the submodel $\M'$ but with a new labeling function $\hat{\eta}$ defined as
              $
                  \hat{\eta}(w') \coloneqq \{x_{w'}\}
              $
              for all $w' \in W'$ with $x \in \eta(w')$.

              Accept if $\hat{\M} \models \EF (x_{w'} \land \EX\EF x_{w'})$ for some $w' \in W'$, else reject.
    \end{itemize}

    Correctness of the first two cases is trivial.
    A path witnessing $\EF x$ or $\EG x$ induces a submodel, where $\AF x$ or $\AG x$ holds, respectively.
    The third case is also quite obvious.
    If $\M' \models \EF\EG x$, then there is a path $\pi$ and a $k$ such that $\M', \pi[k] \models \EG x$.
    Let $\sigma$ be the path witnessing $\M', \pi[k] \models \EG x$, we than have a path $\rho = \pi[1], \dots, \pi[k-1], \rho[1] (= \pi[k]), \rho[2], \dots$ which induces a submodel satisfying $\varphi' = \AF\AG x$.

    For $\AG\AF x$ this approach does not work.
    Consider the following model $\M_0$:
    
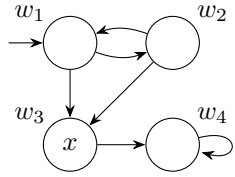
\begin{center}
        \begin{tikzpicture}[scale=0.9]
            \node[world,label={140:$w_1$}] at (0,   1.5) (w1) {};
            \node[world,label={ 40:$w_2$}] at (1.5, 1.5) (w2) {};
            \node[world,label={140:$w_3$}] at (0,     0) (w3) {$x$};
            \node[world,label={ 40:$w_4$}] at (1.5,   0) (w4) {};

            \draw[Stealth-] (w1.west) -- ++(-0.5,0);
            \draw[-Stealth] (w1) to[out =-20, in =-160] (w2);
            \draw[-Stealth] (w1) -- (w3);
            \draw[-Stealth] (w2) to[out = 160, in = 20] (w1);
            \draw[-Stealth] (w2) -- (w3);
            \draw[-Stealth] (w3) -- (w4);
            \draw[-Stealth] (w4) edge[loop right,>=Stealth] (w4);
        \end{tikzpicture}
    \end{center}
    While $\M_0 \models \EG\EF x$ holds, with $\pi = w_1,w_2,w_1,w_2\dots$ as witness, no submodel can satisfy $\AG\AF x$, because all submodel $\M'_0 \subseteq \M_0$ contain $w_4$ and $\M'_0, w_4 \not\models \AF x$, which means $\M'_0 \not\models \AG \AF x$.

    \medskip

    Observer that $\AG\AF x$ implies that \emph{all} path contain infinitely many worlds where $x$ holds.
    Since our models are finite it follows that at least on such world must occur on the path infinitely often.

    We mimic this property in terms of model checking by first constructing another model $\hat{\M}$, where each world $w' \in W'$ labeled with $x$ gets a new and unique label $x_{w'}$.
    Secondly, we construct a formula as a disjunction of $\EF (x_{w'} \land \EX\EF x_{w'})$.
    Notice that this disjunction is true, if and only if the model has at least one cycle containing a world labeled with $x_{w'}$ and thereby a path which contains this worlds infinitely often.
    \begin{align*}
         & \hat{\M} \models \EF (x_{w'} \land \EX\EF x_{w'})                                                                           \\
         & \Leftrightarrow \exists \pi \in \Pi(\hat{\M}) \exists k \geq 1: \hat{\M}, \pi[k] \models x_{w'} \text{ and } \hat{\M}, \pi[k] \models \EX\EF x_{w'}               \\
         & \Leftrightarrow \exists \pi \in \Pi(\hat{\M}) \exists k \geq 1: \hat{\M}, \pi[k] \models x_{w'} \text{ and } \hat{\M}, \pi[k+1] \models \EF x_{w'}                \\
         & \Leftrightarrow \exists \pi \in \Pi(\hat{\M}) \exists k \geq 1: \hat{\M}, \pi[k] \models x_{w'} \text{ and } \exists j \geq k+1 : \hat{\M}, \pi[j] \models x_{w'} \\
         & \Leftrightarrow \exists \pi \in \Pi(\hat{\M}) \exists k \geq 1: \pi[k] = w' \text{ and } \exists j > k : \pi[j] = w'
    \end{align*}
    It then follows that the path
    \begin{equation*}
        \rho = \pi[1], \dots \pi[k-1], \pi[k], \pi[k+1], \dots, \pi[j-1], \pi[k] ( = \pi[j]), \pi[k+1], \dots
    \end{equation*}
    of $\hat{\M}$ induces a satisfying submodel of $\M'$.

    \medskip

    Constructing $\varphi'$ and model checking the first three cases can clearly be done in polynomial time.
    For the fourth case we additionally need to construct a new submodel.
    But since the size of the new model is identical to the old one, this means it can also be done in polynomial time.
    The size of the disjunction is linear in the number of worlds of the submodel.
    Its construction and the model checking can therefore be done in polynomial time.
\end{proof}
Let us illustrate the behaviour of the algorithm with an example.
\begin{example}\label{ex:AFAG_algo}
    Let $\M$ be the Kripke model depicted in \Cref{fig:ex:modelFG}.
    Further let
    $$
        \varphi \coloneqq \AF\AG\AG\AF x.
    $$
    We now call the algorithm from \Cref{th:FG_inDelP} on the input $\langle \M, \varphi, \emptyset \rangle$.

    The first step is to shrink $\varphi$.
    Note that with (1.) from \Cref{lem:FG_trimming} $\varphi \equiv \AF\AG\AF x$ and with (2.) $\AF\AG\AF x \equiv \AG\AF x \eqqcolon \varphi'$.
    So we proceed as follows.

    First we construct the model $\hat{\M}$ (see \Cref{fig:ex:modelFG_relabeled}) and the formula
    $$
        \psi \coloneqq \EF (x_{w_2} \land \EX\EF x_{w_2}) \lor \EF (x_{w_3} \land \EX\EF x_{w_3}).
    $$
    The algorithm then uses a model checking algorithm to test whether $\hat{\M} \models \psi$ holds.
    The model checking algorithm will return true in our case, so our algorithm accepts.

    The model induced by a path witnessing $\psi$ can be seen in \Cref{fig:ex:submodelFG}, also notice that this model obviously satisfies $\varphi'$ and thereby $\varphi$.
    \begin{figure}
        
        \centering
        \begin{subfigure}{0.35\columnwidth}
            \centering
            \begin{tikzpicture}[scale=1]
                \node[world,label={140:$w_1$}] at (0,   1.5) (w1) {};
                \node[world,label={ 40:$w_2$}] at (1.5, 1.5) (w2) {$x$};
                \node[world,label={140:$w_3$}] at (0,     0) (w3) {$x$};
                \node[world,label={ 40:$w_4$}] at (1.5,   0) (w4) {};
                \node[world,label={ 40:$w_5$}] at (3,   1.5) (w5) {};

                \draw[Stealth-] (w1.west) -- ++(-0.5,0);
                \draw[-Stealth] (w1) -- (w2);
                \draw[-Stealth] (w1) -- (w3);
                \draw[-Stealth] (w2) -- (w5);
                \draw[-Stealth] (w5) edge[loop below,>=Stealth] (w5);
                \draw[-Stealth] (w3) -- (w2);
                \draw[-Stealth] (w3) to[out =-20, in =-160] (w4);
                \draw[-Stealth] (w4) to[out = 160, in = 20] (w3);
            \end{tikzpicture}
            \caption{Model $\M$.}
            \label{fig:ex:modelFG}
        \end{subfigure}

        \begin{subfigure}{0.5\columnwidth}
            \centering
            \begin{tikzpicture}[scale=1]
                \node[world,label={140:$w_1$}] at (0,   1.5) (w1) {};
                \node[world,label={ 40:$w_2$}] at (1.5, 1.5) (w2) {$x_{w_2}$};
                \node[world,label={140:$w_3$}] at (0,     0) (w3) {$x_{w_3}$};
                \node[world,label={ 40:$w_4$}] at (1.5,   0) (w4) {};
                \node[world,label={ 40:$w_5$}] at (3,   1.5) (w5) {};

                \draw[Stealth-] (w1.west) -- ++(-0.5,0);
                \draw[-Stealth] (w1) -- (w2);
                \draw[-Stealth] (w1) -- (w3);
                \draw[-Stealth] (w2) -- (w5);
                \draw[-Stealth] (w5) edge[loop below,>=Stealth] (w5);
                \draw[-Stealth] (w3) -- (w2);
                \draw[-Stealth] (w3) to[out =-20, in =-160] (w4);
                \draw[-Stealth] (w4) to[out = 160, in = 20] (w3);
            \end{tikzpicture}
            \caption{Model $\hat{\M}$. Identical frame, but different labels.}
            \label{fig:ex:modelFG_relabeled}
        \end{subfigure}
        \hfill
        \begin{subfigure}{0.35\columnwidth}
            \centering
            \begin{tikzpicture}[scale=1]
                \node[world,label={140:$w_1$}] at (0, 1.5) (w1) {};
                \node[world,label={140:$w_3$}] at (0,   0) (w3) {$x$};
                \node[world,label={ 40:$w_4$}] at (1.5, 0) (w4) {};

                \draw[Stealth-] (w1.west) -- ++(-0.5,0);
                \draw[-Stealth] (w1) -- (w3);
                \draw[-Stealth] (w3) to[out =-20, in =-160] (w4);
                \draw[-Stealth] (w4) to[out = 160, in = 20] (w3);
            \end{tikzpicture}
            \caption{Induced submodel of a path witnessing $\EF (x_{w_3} \land \EX\EF x_{w_3})$.}
            \label{fig:ex:submodelFG}
        \end{subfigure}
        \caption{Original and intermediate Kripke model of \Cref{ex:AFAG_algo} as well as the submodel found by the algorithm, that satisfies $\varphi' = \AG\AF x$.}
    \end{figure}
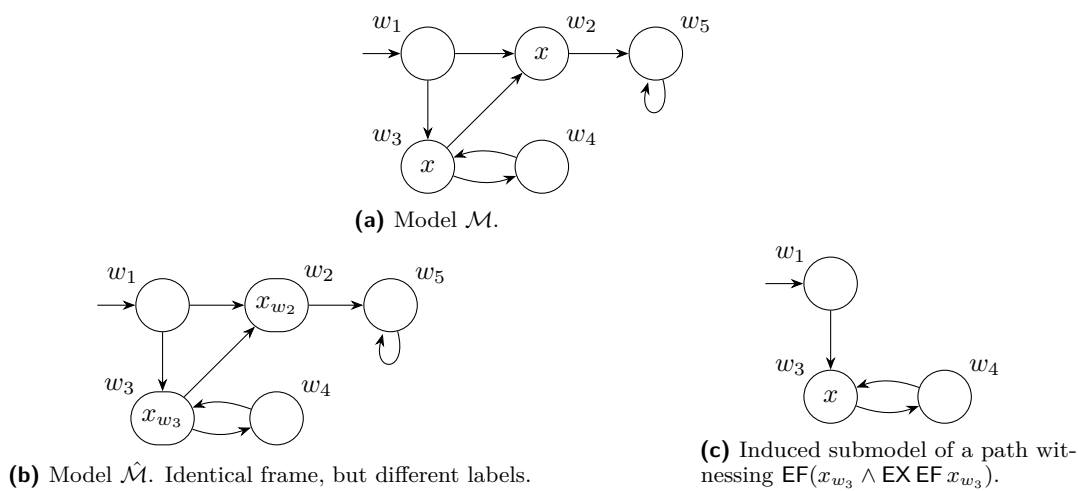
\end{example}

\section{Conclusion and outlook}
In this paper, we have presented a complete study of the submodel enumeration problem for the temporal logic CTL with respect to restrictions on the allowed CTL operators.
We have examined all CTL operator fragments and show $\DelNP$-completeness for every possible fragment in the presence of all Boolean connectors.
This paints a completely negative picture and precludes using the debugging approach as motivated in this setting.
As a silver lining on the horizon, we presented fragments obtained by constraints on  Boolean functions, allowing for fast $\DelP$ algorithms that could be used for bugfix recommendations.
We are currently planning to extend this approach to a complete picture for all Boolean fragments and combinations with CTL operator fragments.
In particular, this leads to a very large number of possible fragments: as a rough estimate, one has to consider seven Boolean fragments, which, combined with ten CTL operators, lead to an astonishing number of $7\cdot 2^{10}=7168$ cases. 
As future work, it would be worthwhile to apply the framework of parameterised complexity~\cite{DBLP:series/mcs/DowneyF99} aiming at more efficient subcases.
Another pressing issue is to investigate the motivated debugging approach using enumeration algorithms in a feasibility study.
Furthermore, submodel enumeration is just one of many possible enumeration problems for CTL.
Other variants worth investigating in this context include (minimal) modifications to $\eta$ instead of, or in addition to, frame modifications.

\bibliography{enumCTL_arXiv}

\end{document}